%% file: main.tex
\documentclass[a4paper,UKenglish,cleveref, autoref, thm-restate]{lipics-v2021}

\input{globals/packages}

\input{globals/macros}

\title{Deciding the Value of Two-Clock Almost Non-Zeno Weighted Timed Games}

\author{Isa {Vialard}}{Max Planck Institute for Software Systems, Saarland Informatics Campus, Germany \and \url{https://isavialard.github.io/home/} }{vialard@mpi-sws.org}{https://orcid.org/0000-0002-7261-9342}{}

\authorrunning{I. Vialard} 

\Copyright{Isa Vialard} 



\category{} 

\relatedversion{} 

\acknowledgements{I want to thank Quentin Guilmant for his help in formalizing the proofs in \cref{sec-app}.}

\nolinenumbers 

\EventEditors{John Q. Open and Joan R. Access}
\EventNoEds{2}
\EventLongTitle{42nd Conference on Very Important Topics (CVIT 2016)}
\EventShortTitle{CVIT 2016}
\EventAcronym{CVIT}
\EventYear{2016}
\EventDate{December 24--27, 2016}
\EventLocation{Little Whinging, United Kingdom}
\EventLogo{}
\SeriesVolume{42}
\ArticleNo{23}

\begin{document}

\maketitle

\begin{abstract}
 The Value Problem for weighted timed games (\wtgs) consists in determining, given a two-player weighted timed game with a reachability objective and a rational threshold, whether or not the value of the game exceeds the threshold.
When restrained to \wtgs with non-negative weight, this problem is known to be undecidable for weighted timed games with three or more clocks, and decidable for $1$-clock \wtgs.
The Value Problem for two-clock non-negative \wtgs, which remained stubbornly open for a decade, was recently shown to be undecidable. In this article, we show that the Value Problem is decidable when considering $2$-clock almost non-Zeno \wtgs.
\end{abstract}
%
%

\input{parts/intro}
\input{parts/def}

\input{parts/unfolding}

\input{parts/regions}

\input{parts/simplification}

\input{parts/valueiteration}

\clearpage

\bibliography{../../BIBLIO/biblio} 

\clearpage

\appendix

\input{parts/appendix}

\end{document}

%% file: globals/packages.tex
\usepackage{mathtools}
\usepackage{xspace}

\usepackage{booktabs}

\usepackage{tikz,xcolor,hyperref}
\usetikzlibrary{positioning, arrows}
\usetikzlibrary{patterns}
\usetikzlibrary{shapes}
\usetikzlibrary{decorations.pathreplacing}
\usetikzlibrary{fadings}
\usetikzlibrary{intersections}

\usepackage{todonotes}
\setuptodonotes{inline}

\usepackage{lineno}


%% file: globals/macros.tex
\newcommand{\openquote}{``}

\newcommand{\ie}{i.e.,~}
\newcommand{\etal}{\textit{et al.}~}

\usetikzlibrary{decorations,decorations.markings} 
\pgfkeys{/tikz/.cd,
    contour distance/.store in=\ContourDistance,
    contour distance=-10pt, 
    contour step/.store in=\ContourStep,
    contour step=1pt,
}

\pgfdeclaredecoration{closed contour}{initial}
{%
\state{initial}[width=\ContourStep,next state=cont] {
    \pgfmoveto{\pgfpoint{\ContourStep}{\ContourDistance}}
    \pgfcoordinate{first}{\pgfpoint{\ContourStep}{\ContourDistance}}
    \pgfpathlineto{\pgfpoint{0.3\pgflinewidth}{\ContourDistance}}
    \pgfcoordinate{lastup}{\pgfpoint{1pt}{\ContourDistance}}
    
  }
  \state{cont}[width=\ContourStep]{
     \pgfmoveto{\pgfpointanchor{lastup}{center}}
     \pgfpathlineto{\pgfpoint{\ContourStep}{\ContourDistance}}
     \pgfcoordinate{lastup}{\pgfpoint{\ContourStep}{\ContourDistance}}
  }
  \state{final}[width=\ContourStep]
  { 
    \pgfmoveto{\pgfpointanchor{lastup}{center}}
    \pgfpathlineto{\pgfpointanchor{first}{center}}
  }
}

\newcommand{\eqdef}{\stackrel{\mathsf{def}}{=}}

\newcommand{\set}[1]{\left\{\; {#1} \;\right\}}

\newcommand{\setof}[2]{\left\{{#1} \;:\; {#2} \right\}}

\newcommand{\Nat}{\mathbb{N}}
\newcommand{\Rel}{\mathbb{Z}}
\newcommand{\Reel}{\mathbb{R}}

\newcommand{\DefineNewOrder}[4]{%
	\expandafter\newcommand\csname#1le\endcsname{\mathrel{#2}_{#4}}%
	\expandafter\newcommand\csname#1ge\endcsname{\mathrel{\reflectbox{\text{$#2$}}}_{#4}}%
	\expandafter\newcommand\csname#1leq\endcsname{\mathrel{#3}_{#4}}%
	\expandafter\newcommand\csname#1geq\endcsname{\mathrel{\text{\reflectbox{$#3$}}}_{#4}}%
	\expandafter\newcommand\csname#1eq\endcsname{\mathrel{\equiv}_{#4}}%
	\expandafter\newcommand\csname#1perp\endcsname{\mathrel{\perp}_{#4}}%
}

\newcommand{\Min}{{Min}\;}

\newcommand{\Value}{\mathsf{Val}}
\newcommand{\reset}[1]{{#1}:=0}
\newcommand{\play}{\textit{play}}

\DeclareMathOperator{\Reg}{Reg}
\newcommand{\Maxsf}{\mathsf{Max}}
\newcommand{\Minsf}{\mathsf{Min}}
\newcommand{\Mamax}{\mathsf{Max}}
\newcommand{\Mimin}{\mathsf{Min}}
\newcommand{\Player}{\mathsf{P}}
\newcommand{\upclock}[1]{X_{#1}^\uparrow}
\newcommand{\weight}{\mathsf{weight}}
\newcommand{\opt}[2]{\mathsf{opt}^{#1}_{#2}}
\newcommand{\Opt}[2]{\mathsf{Opt}^{#1}_{#2}}
\newcommand{\Next}{\mathsf{Leaving}}
\newcommand{\reg}{\mathsf{reg}}
\newcommand{\regadh}[1]{\overline{\mathsf{reg}(#1)}}
\newcommand{\wout}{w_{\mathsf{out}}}
\newcommand{\wtg}{\textsc{wtg}\xspace}

\newcommand{\wtgs}{\textsc{wtg}s\xspace}

\newcommand{\Rbb}{\ensuremath{\mathbb{R}}}

\newcommand{\Ccal}{\ensuremath{\mathcal{C}}}

\newcommand{\Gcal}{\ensuremath{\mathcal{G}}}

\newcommand{\Kcal}{\ensuremath{\mathcal{K}}}

\newcommand{\Rcal}{\ensuremath{\mathcal{R}}}
 
\newcommand{\Tcal}{\ensuremath{\mathcal{T}}}

\newcommand{\Xcal}{\ensuremath{\mathcal{X}}}

\newcommand{\ignore}[1]{}
\newcommand{\tiff}{\textit{iff} }

\if@anglais

\else

\fi

\newcommand{\intxy}[4]{\ensuremath{\mathchoice
    {\left#3\protect\TestGauche{#1},  \protect\TestDroite{#2}\right#4}
    {\left#3\protect\TestGauche{#1},  \protect\TestDroite{#2}\right#4}
    {\left#3\protect\TestGauche{#1},\protect\TestDroite{#2}\right#4}
    {\left#3\protect\TestGauche{#1},\protect\TestDroite{#2}\right#4}
}}  
	\newcommand{\intff}[2]{\intxy{#1}{#2}{[}{]}}
	\newcommand{\intof}[2]{\intxy{#1}{#2}{(}{]}}
	\newcommand{\intfo}[2]{\intxy{#1}{#2}{[}{)}}

\newcommand{\TestGauche}[1]{\ifthenelse{\equal{#1}{}}{\minf}{#1}}
\newcommand{\TestDroite}[1]{\ifthenelse{\equal{#1}{}}{\pinf}{#1}}

\newcommand{\enstq}[2]{\left\{ \left.#1\ \vphantom{#2}\right|\ #2  \right\}}

\newcommand{\pa}[2][9]{%
	\ifthenelse{#1 = 0}{#2}{}%
	\ifthenelse{#1 = 1}{(#2)}{}%
	\ifthenelse{#1 = 2}{\big(#2\big)}{}%
	\ifthenelse{#1 = 3}{\Big(#2\Big)}{}%
	\ifthenelse{#1 = 4}{\bigg(#2\bigg)}{}%
	\ifthenelse{#1 = 5}{\Bigg(#2\Bigg)}{}%
	\ifthenelse{#1 = 9}{\left(#2\right)}{}%
}

%% file: parts/intro.tex
\section{Introduction}

Introduced by Alur and Dill (\cite{Alur94}) in the early 1990s,
a \emph{timed automaton} is an automaton where transitions are limited by time constraints on a set of finite clocks. 
\emph{Weighted timed automata}, also known as priced timed automata, are timed automata with integer costs added to locations and transitions. These costs can be punctual, or linear in terms of time spent in a location.
Timed automata and weighted timed automata are powerful models for real-time systems—for instance task scheduling, controller synthesis, energy-aware systems, etc.

Real-time systems often have to deal with perturbations from an uncontrollable environment (for instance, a user). This can be modelled by \emph{timed games}: timed automata where transitions are divided among two players, $\Minsf$ (the control), who has a reachability objective, and her opponent $\Maxsf$ (the environment).

When adding costs to timed games, we obtain \emph{weighted timed games} (\wtgs): $\Minsf$ now attempts to reach a goal location while minimizing the cost of doing so.

A natural problem on \wtgs is the following: 
Given a \wtg and a threshold $c$, is the infimum of the optimal cost
\footnote{where the optimal cost is the supremum on all possible strategies of $\Maxsf$ of the weight of the path produced by the strategy profile}
 on all strategies of $\Minsf$ less than or equal to $c$?
 This is the \emph{Value Problem}, not to be confused with the \emph{Existence Problem}:
 Given a \wtg and a threshold $c$, does $\Minsf$ have a strategy to reach her goal location with cost at most $c$?

%

In this article, we focus on the Value Problem. While decidable for weighted timed automata, the Value Problem is undecidable for weighted timed games (\cite{Bouyer2015}).

However, one can recover decidability by restricting the number of clocks. Bouyer \etal (\cite{bouyer2006}) establish that the Value problem is decidable for one-clock \wtg with non-negative weights. This decidability result is extended to one-clock \wtg with arbitrary weights in \cite{monmege2024}.
On the other hand, Bouyer \etal (\cite{Bouyer2015}) prove undecidability for \wtgs with non-negative weight and $3$-clocks or more. 
Brihaye \etal (\cite{brihaye2014}) show undecidability of two-clock \wtg with arbitrary weights. Only recently, Guilmant \etal (\cite{two-clock-undec}) proved undecidability of two-clock, time-bounded \wtg with non-negative weights.

Another way to recover decidability is with non-Zeno (or divergence) properties. A \wtg with non-negative weights has a \emph{strictly non-Zeno cost property} when every cycle is of cost at least $1$. Intuitively, this property forbids any \openquote Zeno paradox'' behaviour. Strictly non-Zeno \wtgs can be enfolded into acyclic \wtgs; hence, the Value Problem is decidable (\cite{BCFL-fsttcs04}). This result was generalized to \wtgs with arbitrary weight in \cite{Busatto2017}, for which non-Zenoness becomes \emph{divergence}.
\footnote{A \wtg is \emph{divergent} if every strongly connected component has either cycles of weight in $\intof{-\infty}{-1}$ or 
cycles of weight in $\intfo{1}{\infty}$}

Divergence properties can be weakened into almost-divergence properties. A \wtg with non-negative weight is said to be \emph{almost non-Zeno} (or almost strictly non-Zeno, or almost strongly non-Zeno) if its cycles are of weight $0$, or at least $1$. Bouyer \etal (\cite{Bouyer2015}) establish that the optimal value of such \wtgs is approximable (but still undecidable). Busatto-Gaston \etal (\cite{Busatto2018}) extend this result to \emph{almost divergent}\footnote{A \wtg is \emph{almost divergent} if every strongly connected component has either cycles of weight in $\intof{-\infty}{-1}\cup\{0\}$, or 
cycles of weight in $\{0\}\cup\intfo{1}{\infty}$}
 \wtgs with arbitrary weights.

\input{parts/landscape}

\subsection*{Contributions}

The main theorem of this article is the following: 

\begin{restatable}{theorem}{thmDec}\label{thm-dec}
Given a two-player, turn-based, two-clocks, almost non-Zeno weighted timed game with non-negative integer weights, the
Value Problem is decidable.
\end{restatable}

The proof of \cref{thm-dec} relies on the partial unfolding used in the approximability proof of \cite{Bouyer2015}, and on several techniques to turn a \wtg into an equivalent, simpler game with desirable properties, such as the relaxing of guards, or adding clock resets to every transition.

%% file: parts/landscape.tex
\begin{figure}
		\begin{tabular}{rcccc}
		\toprule
		Number of
		&\multirow{2}{*}{Weights in}
		& \multirow{2}{*}{WTG} 
		& {Almost divergent}
		& Divergent \\
		clocks & & & WTG & WTG\\
		\midrule
		\multirow{2}{*}{$1$}&$\Nat$ 
			& Decidable \cite{bouyer2006} 
			& Decidable 
			& Decidable\\
		&$\Rel$
			& Decidable \cite{monmege2024}
			& Decidable 
			& Decidable\\ 
		\midrule
		\multirow{3}{*}{$2$}& $\Nat$
			& \multirow{2}{*}{Undec.  \cite{two-clock-undec}}
			& \textcolor{blue}{\textbf{Decidable} }
			& \multirow{2}{*}{Decidable}\\
		&&& \textcolor{blue}{\textit{(Our contribution)}}&
		 \\
		&$\Rel$
			&Undec. 	
			& \emph{Decidability open} 	
			& Decidable\\
		\midrule
		\multirow{4}{*}{$\geq 3$}&\multirow{2}{*}{$\Nat$}
			&\multirow{2}{*}{Undec.}	
			&Undec.	\cite{Bouyer2015}
			&\multirow{2}{*}{Decidable \cite{BCFL-fsttcs04}}\\
	&&&\textit{(Approximable)}\cite{Bouyer2015}&\\
	&\multirow{2}{*}{$\Rel$}
		&Undec.	
		&Undec. 
		&\multirow{2}{*}{Decidable \cite{Busatto2017}}\\
	&& \textit{(Non approx.)}\cite{guilmant2024}& \textit{(Approximable)} \cite{Busatto2018}&\\ 
		\end{tabular}

\begin{tikzpicture}[overlay]
\def\x{4}
\def\xx{6.7}
\def\xxx{10.1}
\def\y{4.5}
\def\yy{3.3}
\def\yyy{2.3}
\def\yyyy{0}
\def\z{7}
\def\zz{9.8}
\draw[red,ultra thick,rounded corners] 
(\x,\y) 
-- (\x,\yy)
-- (\xx,\yy)
-- (\xx,\yyy)
-- (\z,\yyy);
\draw[red,ultra thick,rounded corners] 
(\zz,\yyy)
-- (\xxx,\yyy)
-- (\xxx,\yyyy);
\draw[red,ultra thick,dotted] (\xxx,\yyyy)-- (\xxx,\yyyy-.5);

\end{tikzpicture}
\caption{Landscape of WTG decidability and approximability.}
\label{fig-landscape}
\end{figure}

%% file: parts/def.tex
\section{Definitions}

Let $\mathcal{X}$ be a finite set of \textbf{clocks}. \textbf{Clock constraints} (or \textbf{guards}) over $\mathcal{X}$ are
	expressions of the form $x \mathrel{\sim} n$, where $x,y\in\mathcal{X}$ are clocks,
        ${\sim}\in\{<,\leq,=,\geq,>\}$ is a comparison symbol, and
        $n\in\mathbb{N}$ is a natural number.  We write $\mathcal{C}$ to denote
        the set of all clock constraints over $\mathcal{X}$.
A \textbf{valuation} on $\mathcal{X}$ is  a function $\nu:\mathcal{X}\to\mathbb{R}_{\geq 0}$. 
	For $d\in\mathbb{R}_{\geq 0}$ we denote by $\nu+d$ the valuation such
        that, for all clocks $x \in \mathcal{X}$, $(\nu+d)(x)=\nu(x)+d$. Let
        $X \subseteq \mathcal{X}$ be a subset of all clocks. We write
        $\nu[X :=0]$ for the valuation such that, for all clocks $x
        \in X$, $\nu[X :=0](x) = 0$, and $\nu[X:=0](y) = \nu(y)$ for
        all other clocks $y \notin X$.
	For $C \subseteq \mathcal{C}$ a set of clock constraints over $\mathcal{X}$, we say that
        the valuation $\nu$ \textbf{satisfies} $C$, denoted $\nu \models
        C$, if and only if all the comparisons in $C$ hold
        when replacing each clock $x$ by its corresponding value $\nu(x)$.

\begin{definition}
	A \textbf{(turn-based) weighted timed game} is given by a
        tuple $\mathcal{G}=$\linebreak $(L_\mathsf{Min}, L_\mathsf{Max},G,\mathcal{X},T,w)$, where:
	\begin{itemize}
        \item $L_\mathsf{Min}$ and $L_\mathsf{Max}$ are the (disjoint)
          sets of \textbf{locations} belonging to Players $\mathsf{Min}$ and
          $\mathsf{Max}$ respectively; we let $L = L_\mathsf{Min} \cup
          L_\mathsf{Max}$ denote the set of all locations. (In drawings,
          locations belonging to $\mathsf{Min}$ are depicted by blue
          circles, and those belonging to $\mathsf{Max}$ are depicted
          by red squares.)
          \item $G\subseteq L_{\mathsf{Min}}$ are the \textbf{goal locations}.
		\item $\mathcal{X}$ is a set of clocks.
		\item $T\subseteq (L \setminus G) \times2^\mathcal{C}\times 2^\mathcal{X}\times
                  L$ is a set of \textbf{(discrete) transitions}. A
                  transition $\ell \xrightarrow{C,X} \ell'$
                  enables moving from location $\ell$ to
                  location $\ell'$, provided all clock constraints in
                  $C$ are satisfied, and afterwards resetting all clocks in $X$
                  to zero. 
		\item $w:(L \setminus G) \cup T\to \mathbb{Z}$ is a \textbf{weight function}.
                \end{itemize}
In the above, we assume that all data (set of locations, set of
clocks, set of transitions, set of clock constraints) are finite.             
\end{definition}

          Let $\mathcal{G}= (L_\mathsf{Min}, L_\mathsf{Max},G,\mathcal{X},T,w)$ be
          a \wtg.  A \textbf{configuration} over
          $\mathcal{G}$ is a pair $(\ell,\nu)$, where $\ell\in L$ and $\nu$
          is a valuation on $\mathcal{X}$. Let $d \in \mathbb{R}_{\geq 0}$ be a
          \textbf{delay} and $t = \ell \xrightarrow{C,X} \ell' \in T$
          be a discrete transition. One then has a \textbf{delayed
            transition} (or simply a \textbf{transition} if the
          context is clear)
          $(\ell,\nu) \xrightarrow{d,t} (\ell',\nu')$ provided that
          $\nu+d \models C$ and $\nu' = (\nu+d)[X := 0]$. Intuitively,
          control remains in location $\ell$ for $d$ time units, after
          which it transitions to location $\ell'$,
          resetting all the clocks in $X$ to zero in the process.  The
          \textbf{weight} of such a delayed transition is
          $d \cdot w(\ell) + w(t)$, taking account both of the time
          spent in $\ell$ as well as the weight of the discrete
          transition $t$.

          As noted in~\cite{BusattoGastonMR23},
          without loss of generality one can assume that no
          configuration (other than those associated with goal
          locations) is deadlocked; in other words, for any location
          $\ell \in L\setminus G$ and valuation
          $\nu \in \mathbb{R}_{\geq 0}^{\mathcal{X}}$, there exists
          $d \in \mathbb{R}_{\geq 0}$ and $t \in T$ such that
          $(\ell,\nu) \xrightarrow{d,t} (\ell',\nu')$.\footnote{This
            can be achieved by adding unguarded transitions to a
            sink location for all locations controlled by $\mathsf{Min}$ and
            unguarded transitions to a goal location for the ones
            controlled by $\mathsf{Max}$.}

          Let $k\in\mathbb{N}$.  A \textbf{run} $\rho$ of length $k$ over $\mathcal{G}$ from a given configuration
          $(\ell_0,\nu_0)$ is a sequence of matching delayed
          transitions, as follows:
          \[ \rho = (\ell_0,\nu_0) \xrightarrow{d_0,t_0}
            (\ell_1,\nu_1) \xrightarrow{d_1,t_1} \cdots
            \xrightarrow{d_{k-1},t_{k-1}}  (\ell_k,\nu_k) \, . \]
          The \textbf{weight} of $\rho$ is the
          cumulative weight of the underlying delayed transitions:
          \[ \mathsf{weight}(\rho) = \sum_{i=0}^{k-1} (d_i\cdot
            w(\ell_i) + w(t_i)) \, . \]
          An infinite run $\rho$ is defined in the obvious way;
          however, since no goal location is ever reached, its weight is defined to be infinite:
          $\mathsf{weight}(\rho)=+\infty$.

          A run is \textbf{maximal} if it is either infinite or cannot be extended
          further. Thanks to our deadlock-freedom assumption, finite maximal
          runs must end in a goal location. We refer to maximal runs
          as \textbf{plays}.

          We now define the notion of \textbf{strategy}. Recall
          that locations of $\mathcal{G}$ are partitioned into sets $L_{\mathsf{Min}}$
          and $L_{\mathsf{Max}}$, belonging respectively to Players
          $\mathsf{Min}$ and $\mathsf{Max}$.
          Let Player $\mathsf{P} \in
          \{\mathsf{Min},\mathsf{Max}\}$, and 
write $\mathcal{FR}_{\mathcal{G}}^{\mathsf{P}}$
          to denote the collection of all non-maximal finite runs of $\mathcal{G}$ 
          ending in a location belonging to
          Player $\mathsf{P}$. A \textbf{strategy}
         for Player $\mathsf{P}$ is a mapping
           $\sigma_{\mathsf{P}} : \mathcal{FR}_{\mathcal{G}}^{\mathsf{P}}
           \to \mathbb{R}_{\geq 0} \times T$ such that for all 
           finite runs $\rho \in \mathcal{FR}_{\mathcal{G}}^{\mathsf{P}}$
           ending in configuration $(\ell,\nu)$ with
           $\ell \in L_{\mathsf{P}}$, the delayed transition
           $(\ell,\nu) \xrightarrow{d,t} (\ell',\nu')$ is valid,
           where $\sigma_{\mathsf{P}}(\rho)=(d,t)$ and $(\ell',\nu')$
           is some configuration (uniquely determined by
           $\sigma_{\mathsf{P}}(\rho)$ and $\nu$).

Let us fix a starting configuration $(\ell_0,\nu_0)$, and let
$\sigma_{\mathsf{Min}}$ and $\sigma_{\mathsf{Max}}$ be strategies for
Players $\mathsf{Min}$ and $\mathsf{Max}$ respectively (one speaks of
a \emph{strategy profile}). Let us denote by
$\mathsf{play}_\mathcal{G}((\ell_0,\nu_0),\sigma_{\mathsf{Min}},\sigma_{\mathsf{Max}})$
the unique maximal run starting from configuration
$(\ell_0,\nu_0)$ and unfolding according to the strategy profile
$(\sigma_{\mathsf{Min}},\sigma_{\mathsf{Max}})$: in other words,
for every strict finite prefix $\rho$ of
$\mathsf{play}_\mathcal{G}((\ell_0,\nu_0),\sigma_{\mathsf{Min}},\sigma_{\mathsf{Max}})$
in $\mathcal{FR}_{\mathcal{G}}^{\mathsf{P}}$, the delayed transition
immediately following $\rho$ in
$\mathsf{play}_\mathcal{G}((\ell_0,\nu_0),\sigma_{\mathsf{Min}},\sigma_{\mathsf{Max}})$
is labelled with $\sigma_{\mathsf{P}}(\rho)$.

Recall that the objective of Player $\mathsf{Min}$ is to reach a goal
location through a play whose weight is as small possible. Player
$\mathsf{Max}$ has an opposite objective, trying to avoid goal
locations, and, if not possible, to maximise the cumulative weight of
any attendant play. This gives rise to the following two symmetrical definitions:
\begin{align*}
\overline{\mathsf{Val}}_\mathcal{G}(\ell_0,\nu_0) &=
  \inf_{\sigma_{\mathsf{Min}}} \left\{ \sup_{\sigma_\mathsf{Max}}
                                              \left\{
                                              \mathsf{weight}(\mathsf{play}_\mathcal{G}((\ell_0,\nu_0),\sigma_{\mathsf{Min}},\sigma_{\mathsf{Max}}))\right\}
  \right\} \mbox{\ and}\\
\underline{\mathsf{Val}}_\mathcal{G}(\ell_0,\nu_0) &=
  \sup_{\sigma_{\mathsf{Max}}} \left\{\inf_{\sigma_\mathsf{Min}}
                                               \left\{\mathsf{weight}(\mathsf{play}_\mathcal{G}((\ell_0,\nu_0),\sigma_{\mathsf{Min}},\sigma_{\mathsf{Max}}))
                                               \right\} \right\} \, .
\end{align*}

$\overline{\mathsf{Val}}_\mathcal{G}(\ell_0,\nu_0)$ represents the smallest 
possible weight that Player $\mathsf{Min}$ can possibly achieve,
starting from configuration $(\ell_0,\nu_0)$,
against best play from Player $\mathsf{Max}$, and conversely for
$\overline{\mathsf{Val}}_\mathcal{G}(\ell_0,\nu_0)$: the latter represents
the largest possible weight that Player $\mathsf{Max}$ can enforce,
against best play from Player $\mathsf{Min}$.\footnote{Technically
  speaking, these values may not be literally achievable; however
  given any $\varepsilon > 0$, both players are guaranteed to have
  strategies that can take them to within $\varepsilon$ of the optimal value.}
As noted in~\cite{BusattoGastonMR23}, turned-based
\wtgs are \emph{determined}, and therefore
$\overline{\mathsf{Val}}_\mathcal{G}(\ell_0,\nu_0) =
\underline{\mathsf{Val}}_\mathcal{G}(\ell_0,\nu_0)$ for any starting
configuration $(\ell_0,\nu_0)$; we denote this common value by
$\mathsf{Val}_\mathcal{G}(\ell_0,\nu_0)$.

\begin{remark}
	Note that $\Value_\mathcal{G}(\ell_0,\nu_0)$ can take on real numbers, or either
	of the values $-\infty$ and $+\infty$.
	However, since reachability is decidable in timed games, it is decidable whether $\Value_\mathcal{G}(\ell_0,\nu_0)=+\infty$ or not.
\end{remark}

\textbf{In the rest of this article, every weighted timed game is turn-based, with non-negative weights, of value in $\mathbb{R}$}.

%% file: parts/unfolding.tex
\section{Unfolding Almost Non-Zeno Weighted Timed Games}

In this section, we assume familiarity with the region construction (see \cite{Alur94}). We denote by $\Rcal(\Gcal)$ the region automaton associated with a \wtg $\Gcal$. In $\Rcal(\Gcal)$, every location $\ell$ is assigned a unique region $\reg(\ell)$ of accessible valuations.
For more details, see \cref{sec-region-trimmed}.

Bouyer \etal (\cite{Bouyer2015}) showed that, even though the Value Problem is undecidable for \wtg with non-negative weight and $3$ clocks or more, it is approximable in the subclass of almost non-Zeno \wtg.
In this section, we use the structure of their proof of approximability to prove decidability for almost non-Zeno \wtgs with $2$ clocks.

\begin{definition}[Almost non-Zeno \wtg]
A \wtg $\Gcal$ is \textbf{almost non-Zeno} if there exists $\kappa >0$ such that for any finite run $\rho$ in $\Gcal$ that follows a region cycle of $\Rcal(\Gcal)$, $\weight(\rho)=0$ or $\geq\kappa$.
\end{definition}
\begin{remark}
 It is decidable whether a weighted
timed game is almost non-Zeno or not (by enumerating all simple cycles in the
corner-point abstraction of $\Gcal$, see \cite{BBL2008}).
\end{remark}

In an acyclic \wtg, the value is decidable and can be computed  from the target locations up to the initial location, by computing for each node $\ell$ a function $W_\ell:\reg(\ell)\to\Rbb$ which assigns to a valuation $\nu\in \reg(\ell)$ the optimal weight 
$\Value_\Gcal(\ell,\nu)$.
By construction, every $W_\ell$ is a piecewise-linear function.

Intuitively, we will unfold cycles of weight $\geq 1$ to obtain a ``tree-like'' \wtg where only cycles of weight $0$ are left; we will deal with them separately.

\paragraph*{Semi-unfolding}
 For any trimmed region \wtg $\Gcal$, let $\tilde{\Gcal}$ be the semi-unfolded \wtg built from $\Rcal(\Gcal)$
  in \cite{Bouyer2015}:
 
   First color in green every location and edge that are part of a cycle of weight $0$. 
   Observe that you can modify any \wtg such that any green location has weight $0$\footnote{\cite{Bouyer2015} make a similar observation, but their construction implies adding a clock.}: in a trimmed region \wtg, if a location $\ell\in L_\Player$ of weight $p>0$ is part of a cycle of weight $0$, then there exists an outgoing transition from $\ell$ with a guard $x=0$ for some clock $x$. Therefore, as in \cref{fig-weight-zero}, one can add a location $\ell_0$ of weight $0$ in $L_\Player$ such that:
   \begin{itemize}
   \item every transition arriving in $\ell$ arrives in $\ell_0$ instead.
   \item every \emph{green} transition leaving $\ell$ leaves $\ell_0$ instead.
   \item there is a transition $\ell_0 \overset{x=0}{\rightarrow} \ell$. 
   \end{itemize}
   Thus let us assume that every green location has weight $0$.
\begin{figure}
\centering
\includegraphics[scale=.75]{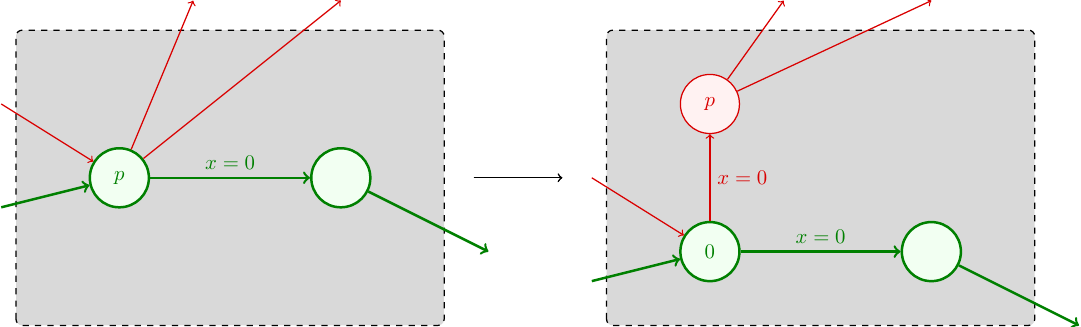}
\caption{How to ensure that every green location has weight $0$.
Thick green transitions and locations are part of cycles of weight $0$, locations labeled with weight $0$ or $p$ belong to the same player.}
\label{fig-weight-zero}
\end{figure} 
We define $\Kcal$ the \textbf{kernel} of $\Gcal$ as the restriction of $\Rcal(\Gcal)$ to fully-green strongly
connected components. Edges that leave $\Kcal$ are called the \textbf{output edges} of $\Kcal$.

Then we partially unfold $\Rcal(\Gcal)$ into a finite tree structure $\Tcal(\Gcal)$ :
starting from the initial location $i$ as a root, we follow every possible path in $\Gcal$, with a node for each time we visit a (non kernel) location, as to avoid creating cycles.
However
when along a branch we enter the kernel in some location $\ell$, we create a node $\Kcal_\ell$ instead of $\ell$, and for each output edge $t$ of $\Kcal$ accessible from $\ell$, with $t$ leading to a location $\ell'$, let $\ell'$ (or $\Kcal_\ell'$ if $\ell'\in\Kcal$) be a child of $\Kcal_\ell$, and continue to unfold from there.

We stop unfolding when, along any branch, a location or edge with positive weight of $\Rcal(\Gcal)$ is visited at
least $W/\kappa + 2$ times, where $W$ is an upper bound on the value of $\Gcal$.\footnote{obtained by using the corner-point abstraction, or considering a memoryless region-uniform strategy for $\Minsf$.}

To obtain $\tilde{\Gcal}$ from $\Tcal(\Gcal)$, replace each node $\Kcal_\ell$ by a copy of the strongly connected component of $\Kcal$ that contains $\ell$
(see \cite{Bouyer2015} for the formal construction). Then $\Value_{\tilde{\Gcal}}(i,\nu)=\Value_\Gcal(i,\nu)$ for any $\nu\in\reg(i)$.


In the partially unfolded games $\tilde{\Gcal}$ with three clocks or more, the cause of undecidability is inside the kernel nodes.
Hence, in \cite{Bouyer2015}, the approximation happens in the kernels. However, with only two clocks, the value is decidable in kernel weighted timed games:

\begin{definition}[Kernel weighted timed games]
A \textbf{kernel weighted timed game} $\Gcal$ is a $[0,1]$-\wtg 
$(L_\Minsf, L_\Maxsf,G,\mathcal{X},T,w,\wout)$ where every location or transition is of weight $0$, and each target location $\ell\in G$ has an output weight function $\wout(\ell,\cdot):\reg(\ell)\to \Reel$ which is continuous and piecewise linear.
In later notations, we omit $w$.
\end{definition}

\begin{restatable}{theorem}{thmKernel}
\label{lem-linear-kernel}
For any two-clock kernel \wtg $\Gcal$, for  any location $i\in\Gcal$, $W_i$ is a continuous piecewise-linear function which can be computed through the value iteration algorithm.
\end{restatable}

This is the main technical result of this articl, which we will prove in \cref{sec-value-kernel}. Let us show first how \cref{lem-linear-kernel} entails value decidability of the partial enfolding $\tilde{\Gcal}$:

\begin{lemma}
For every node $n$ in the tree $\Tcal(\Gcal)$, one can compute $W_n$ a continuous piecewise linear function such that $W_n: \nu\mapsto \Value_{\tilde{\Gcal}}(\ell,\nu)$, where $\ell$ is either $n$ if $n\not\in\Kcal$, or the entrance location of $n=\Kcal_\ell$.
\end{lemma}
\begin{proof}
In the tree structure of $\tilde{\Gcal}$, consider a node $n$:
if $n$ is a leaf, then $n \in G$. Thus let $W_n$ be the constant null function.
Now consider that $n$ is not a leaf, and by induction hypothesis assume that for every child $n'$ of $n$, $W_n'$ is continuous and piecewise linear.
If $n=\ell \not\in\Kcal$ then 
\[W_n:\nu \mapsto \underset{\nu+\delta\models C}{\underset{n \overset{C,X}{\rightarrow} n' \in \Tcal(\Gcal)}{\inf/\sup}}
W_{n'}(\nu+\delta[\reset{X}])\,.\] 
Thus by induction $W_n$ is also continuous and piecewise linear. 

Otherwise, $n=\Kcal_\ell$ for some $\ell\in \Kcal$. Let $K$ be the SCC containing $\ell$, and $T_{\mathsf{out}}$ the output edges leaving from $K$. Consider the kernel \wtg
$K_\ell=(L'_\Minsf,L'_\Maxsf ,G',\mathcal{X},T',\wout)$ where
\begin{itemize}
\item $G'=\setof{\ell_t}{t \in T_{out}}$.
\item For every $\Player$, $L'_\Player = K \cap L_\Player$.
\item $T'=T_{|K\times K}\cup \set{\ell'\overset{C,X}{\to} \ell_t \,|\, t:\ell'\overset{C,X}{\to} \ell'' \in T_{out}}$
\item for every $t:\ell'\overset{C,X}{\to} \ell'' \in T_{\mathsf{out}}$, $\wout(\ell_t,\cdot):\nu\mapsto W_{\ell''}(\nu)+w(t)$, which is piecewise linear by induction hypothesis.
\item $w'$ maps to $0$ always.
\end{itemize} 
Then $W_n=\Value_{K_\ell}(\ell,\cdot)$, which is piecewise linear according to \cref{lem-linear-kernel}.
\end{proof}

This is sufficient to conclude the proof of \cref{thm-dec}.

\thmDec*

%% file: parts/regions.tex
\section{Simplifying Transformations of Kernel Games}
\label{sec-all-simpl}

Before proving \cref{lem-linear-kernel}, let us apply some useful simplifying transformations that preserve the value. These transformations happen in four steps:
\begin{description}
\item[Step $1$:]  Transform a WTG into a WTG with clock values in $\intfo01$.
\item[Step $2$:] Transform a \wtg into a region trimmed \wtg, \ie a \wtg where to every location is assigned a region, and without any \openquote useless'' transition, or guard on transition.
\item[Step $3$:] Transform a trimmed region kernel \wtg by relaxing every strict guard into a strict-or-equal guard.
\item[Step $4$:] Transform a relaxed trimmed region kernel \wtg such that every transition resets at least one clock.
\end{description}

\paragraph*{Commentaries:}
Step $1$ only serves to lighten notations in the rest of this article. In terms of state complexity, the transformations $1+2$ increase the number of locations as much as the classical region construction.

Trimming a \wtg in step $2$ is necessary: without it, Step $3$ would create pathological cases where relaxing some guards would allow a player to take transitions that would have been unreachable in the original \wtg.

Relaxing guards in step $3$ is a technique that has merits of its own outside of the scope of the proof. For instance, the value of a \wtg might not be reached by an optimal strategy for $\Minsf$, but could be the infimum produced by a set of strategies $\epsilon$-close to an optimal strategy in the relaxed \wtg.

Relaxing guard is also a prerequisite for step $4$, which is the most important of all four steps. In a two-clock \wtg, resetting at least one clock in each transition allows us to consider only regions in one dimension. Reducing a two-dimension problem to a one-dimension one is key to the termination argument in \cref{sec-value-kernel}.

\subsection{Restraining Clock Values to [0,1)}
\label{sec-zero-one}

Before presenting the well-known notions of regions and region \wtgs, let us first restrict the setting to 
$\intfo01$-\wtgs, which will simplify the region notations.

\begin{definition}
A \textbf{[0,1)-}\wtg is a weighted timed game where for every reachable configuration $(\ell,\nu)$, $0\leq \nu(x)<1$ for all clock $x$.
\end{definition}

\begin{lemma}
\label{lem-zero-one}
For any \wtg $\Gcal$,
there is an equivalent $\intfo01$-\wtg $\Gcal'$.  
\end{lemma}
\begin{proof}
See \cite{bouyer2006}, Proposition 2 for a detailed proof. Their proof is for $1$-clock \wtg, however, the construction can easily be generalized to any number of clocks.

The intuition of the construction is that the information of the integer parts of the clock can be contained in the locations, while the clocks keep track only of the fractional part:
Let $\mathcal{G}=(L_\Minsf,L_\Maxsf,G,\mathcal{X},T,w)$ and let us build $\mathcal{G'}=(L'_\Minsf,L'_\Maxsf,G,\mathcal{X},T',w')$.
First, w.l.o.g., let us assume that all clocks are bounded by an integer $M$ (\cite{behrmann2001}). Then, for $\Player\in\set{\Minsf,\Maxsf}$, let $L_\Player'=L_\Player \times M^{|\Xcal|}$, and define $T'$ and $w'$ such that a valuation $\nu=(x_1,\dots,x_{|\Xcal|})\in \intfo01^\Xcal$ in a location $(\ell,n_1,\dots, n_{|\Xcal|})$ in $\Gcal'$ is equivalent to a valuation $(n_1+x_1,\dots,n_{|\Xcal|}+x_{|\Xcal|})$ in location $\ell$ in $\Gcal$. 
Note that every transition of $\Gcal'$, with a guard $x=1$ for some clock $x$, resets $x$.
\end{proof}

\subsection{Regions and Region Trimmed Games}
\label{sec-region-trimmed}

\begin{definition}
	Let $\mathcal X$ be a set of clocks in a $\intfo01$-\wtg. 
	A \textbf{region} over $\mathcal X$ is a tuple $r=(X_0,\dots,X_p,X_{=1})$ such that $X_i\neq \emptyset$ for all $1\leq i \leq p$, and
$\{X_0,\dots,X_p,X_{=1}\}$ is a partition of $\Xcal$: $\Xcal=\biguplus_{i=0}^pX_i$

	We denote by $\Reg_\mathcal X$ the set of regions over $\mathcal X$.
	A valuation $\nu$ is said to \textbf{belong} to the region $r$, denoted by $\nu\sqsubset r$, whenever
	\begin{itemize}
		\item $\forall x\in\mathcal X$, $\nu(x)=0\Leftrightarrow x\in X_0$,
		\item $\forall x\in\mathcal X$, $\nu(x)=1\Leftrightarrow x\in X_{=1}$,
		\item $\forall x,y\in\mathcal X$, $\ \nu(x)< \nu(y)<1\Leftrightarrow \exists i,j\in\{0,\dots, p\}\text{ s.t. } i<j\wedge x\in X_i\wedge y\in X_j$.
	\end{itemize}

For $r=(X_0,\dots, X_p, X_{=1})$ a region and $X$ a non-empty subset of $\Xcal$, we denote by $r[X:=0]$ the region $(X_0\cup X,X_1\setminus X,\dots, X_p\setminus X, X_{=1}\setminus X)$. In other words, $r[X:=0]$ is the region such that if $\nu$ belongs to $r$ then $\nu[X:=0]$ belongs to $r[X:=0]$.  A \textbf{$\intfo01$-region} is a region $(X_0,\dots, X_p, X_{=1})$ where $X_{=1}=\emptyset$. Let us abuse notation and denote $r$ by $(X_0,\dots, X_p)$. We denote by $\Reg^<_\mathcal X$ the set  $\intfo01$-regions over $\Xcal$.

A \textbf{time-successor} of a region $r$, with $r=(X_0,\dots, X_p)\in \Reg^<_\mathcal X$  is a region $r'=(X'_0,\dots, X'_{p'}, X'_{=1})$ such that either $r'=r$; or $X_0'=\emptyset$ and 
$X'_i=X_{i-1}$ for $1\leq i \leq p$, and either $X_p=X'_{=1}$ or $X_p=X'_{p+1}$ (and  then $X'_{=1}=\emptyset$).

\end{definition}

We often abuse notation and write $r$ for the set of valuations $\nu\sqsubset r$ it represents.

\begin{definition}[Region \wtg \cite{BBL2008, Busatto2018}]
Let $\mathcal G = (L_\mathsf{Min}, L_\mathsf{Max},G,\mathcal{X},T,w)$ be a $\intfo01$-\wtg.
The \textbf{region} \wtg $\mathcal R(\mathcal G)$ of $\Gcal$ is the $\intfo01$-\wtg $\mathcal R(\mathcal G)=(L'_\mathsf{Min}, L'_\mathsf{Max},G',\mathcal{X},T',w')$ with
\begin{itemize}
	\item $L'_\Player=L_\Player\times\Reg^<_\mathcal X$
	for $\mathsf P\in\{\Minsf,\Maxsf\}$.
	\item $G'=G\times\Reg^<_\mathcal X$.
	\item For every $r=(X_0,\dots X_p)\in \Reg^<_\Xcal$, for every $r'=(X'_0,\dots, X'_{p'}, X'_{=1})\in \Reg_\mathcal X$ a time-successor of $r$,
	if $\ell \xrightarrow{C,X} \ell' \in T$ then
	$(\ell,r) \xrightarrow{C\cup C(r'),X} (\ell',r'[X:=0])$,
with $$C(r')=\setof{(x=0)}{x\in X'_0}\cup \setof{(x=1)}{x\in X'_{=1}}\cup \setof{(0<x<1)}{x\in X'_i, 1\leq i\leq p}.$$
	\item For $\ell\in L_\Minsf\cup L_\Maxsf$ and $r\in\Reg^<_\mathcal X$, $w'(\ell,r)=w(\ell)$.
	\item For $t=(\ell,r) \xrightarrow{C\cup C(r'),X} (\ell',r'[X:=0])\in T'$, 
	$w'(t)=w\left(\ell \xrightarrow{C,X} \ell'\right)$.
\end{itemize}
\end{definition}

While applying simplifying transformations to $\Rcal(\Gcal)$, we wish to preserve the \openquote one-region-per-location'' property. Thus let us formally define what is \emph{a} region \wtg (as opposed to \emph{the} region \wtg).
A \wtg $\Gcal$ is \textbf{a region weighted timed game} if there is a region assignment $\reg: L\to\Reg^<_\mathcal X$, such that
for any transition $t:\ell\xrightarrow{C,X} \ell'$, the valuations $\nu+\delta$ with $\nu\sqsubset \reg(\ell)$ and $\delta\geq 0$ that satisfy $C$ are contained in a unique region $r$, such that $r[X:=0]=\reg(\ell')$. Furthermore, for any initial configuration $(i,\nu)$, we require $\nu\sqsubset \reg(i)$.

	
	Obviously $\Rcal(\Gcal)$ is a region \wtg for any $\intfo01$-\wtg $\Gcal$.
	For any location $\ell$ in a region \wtg, let $\upclock{\ell}\eqdef X_p$ with $\reg(\ell)=(X_0,\dots X_p)$.

	Let us now \openquote trim" $\Rcal(\Gcal)$, \ie delete every useless transition, and every useless guard on transitions :

	\begin{definition}[Trimmed region \wtg]
	
	A  region $\intfo01$-\wtg $\Gcal$ is \textbf{trimmed} if for any transition $t:\ell\xrightarrow{C,X} \ell'$  and any region $r=\reg(\ell)$ in $\Gcal$,
	\begin{itemize}
		\item for any valuation $\nu\sqsubset r$ there exist some $\delta\geq0$ such that $\nu+\delta\models C$.
		\item for any $c\in C$, there exists a valuation $\nu\sqsubset r$ and some $\delta\geq 0$ such that $\nu+\delta$ is in $\intff01^\Xcal$ and
		$\nu+\delta\not\models c$.
	\end{itemize}
\end{definition}

In other words, there is no inaccessible transitions from any tuple location-region. Furthermore, there is no unnecessary clauses in $C$ (the ones that are always verified from the region). Removing inaccessible transitions and unnecessary clauses can always be done from any region \wtg without change in value.

Since every $\intfo01$-\wtg $\Gcal$ is equivalent to the region $\intfo01$-\wtg $\Rcal(\Gcal)$, and every region \wtg is equivalent to a trimmed region \wtg, we can always assume $\Gcal$ to be a trimmed region \wtg.

\begin{observation}
\label{lem-trimmed}
	For $\Gcal$ a trimmed region $\intfo01$-\wtg, for any transition $t:\ell \xrightarrow{C,X} \ell'$, with $\reg(\ell)=(X_0,\dots, X_p)$:
	\begin{itemize}
		\item if $(y=0)\in C$ or $(y>0)\in C$ for some clock $y$ then $y\in X_0$, in other words $y$ is $0$ on the whole region $r$.
		\item if $(y=1)\in C$ for some clock $y$, then $y\in X_p$ in other words $y$ is one of the clocks with largest value in $r$.
		\item there cannot be both $x=0$ and $y=1$ in $C$ for any two clocks $x,y$.
	\end{itemize}
\end{observation}
\begin{proof}
	For the first point, notice that if $y$ was not $0$, then either the transition or the clause would have been trimmed. For the second one, if $y\notin X_p$ then upon taking the transition with condition $y=1$ there is a clock $x\in X_p$ such that $x>y=1$.
	For the third point, for any valuation in $r$, there are no $\delta$ such that $\nu + \delta$ satisfies both clause.
\end{proof}

%% file: parts/simplification.tex
\subsection{Relaxing Strict Guards}
\label{sec-relax}
A kernel \wtg can easily be transformed into a kernel \wtg without strict guards, without change in value.

\begin{restatable}{definition}{defAdh}
Let $r=(X_0,\dots,X_p)$. Then the \textbf{adherence} or $r$, denoted by $\overline{r}$, is the set of regions 
of the form $(Y_0,\dots,Y_{p'},Y_{=1})$ with $p'\leq p$ and $\iota: \intff0{p'+1} \to \intff0p$ strictly increasing such that $\iota(p'+1)=p$ and $Y_0= X_0 \cup \dots \cup X_{\iota{0}}$ and $Y_i= X_{\iota(i-1)+1}\cup \dots \cup X_{\iota(i)}$ for all $1\leq i \leq p'$ and $Y_{=1}= X_{\iota(p')+1}\cup \dots \cup X_{\iota(p'+1)}$.
\end{restatable}
Let us abuse notation and write $\nu\sqsubset \overline{r}$ when $\nu\sqsubset {r'}$ for  $r'\in \overline{r}$.

\begin{restatable}{lemma}{thmspicy}
\label{thm-sim-spicy}
Let $\Gcal_\prec=(L_\Minsf^\prec, L_\Maxsf^\prec,G,\Xcal,T^\prec, \wout^\prec)$ be a trimmed region kernel \wtg. Let $\Gcal_\preceq = (L_\Minsf^\preceq, L_\Maxsf^\preceq,G,\Xcal,T^\preceq,\wout^\preceq)$ be a copy of $\Gcal_\prec$ where 
\begin{itemize}
\item every guard has been relaxed, i.e., every guard of the form $x>0$ and $x<1$ have been replaced by $x\geq 0$ or $x\leq 1$, respectively, for all $x\in\Xcal$.
\item  For any $\ell\in G$, the output function ${\wout}^\preceq(\ell,\cdot)$ is $\wout^\prec(\ell,\cdot)$ extended continuously to $\regadh{\ell}$ in $\Gcal_\preceq$.
\end{itemize}
Then $\Value_{\Gcal_\prec}=\Value_{\Gcal_\preceq}$.
\end{restatable}

The proof of this theorem relies on a bisimulation argument that is detailed in the appendix.

Note that $\Gcal_\preceq$  is not a $\intfo01$-\wtg, but a $\intff01$-\wtg, i.e., every accessible valuation is in $\intff01^\Xcal$.
Furthermore, $\Gcal_\preceq$ is not a region \wtg:


\begin{restatable}{definition}{relaxedreg}
\label{relaxedreg}
A \textbf{relaxed region} \wtg is a $\intff01$-\wtg without strict guard where to each location $\ell$ is assigned a region $\reg(\ell)$:
for any transition $t:\ell\xrightarrow{C,X} \ell'$, the valuations $\nu+\delta$, with $\nu\sqsubset \regadh{\ell)}$ and $\delta\geq 0$, that satisfy $C$ are contained in $\overline{r}$ for a unique region $r$, such that $r[X:=0]=\reg(\ell')$. Furthermore,
 initial configuration $(i,\nu)$ must verify $\nu\in\regadh{i}$.
A \textbf{relaxed trimmed region} \wtg is a relaxed region \wtg
if for any transition $t:\ell\xrightarrow{C,X} \ell'$  and any region $r=\reg(\ell)$ in $\Gcal$,
	\begin{itemize}
		\item for any valuation $\nu\sqsubset \overline{r}$ there exist some $\delta\geq0$ such that $\nu+\delta\models C$.
		\item for any $c\in C$, there exists a valuation $\nu\sqsubset \overline{r}$ and some $\delta\geq 0$ such that $\nu+\delta$ is in $\intff01^\Xcal$ and
		$\nu+\delta\not\models c$.
	\end{itemize}
\end{restatable}

\begin{lemma}
Let $\Gcal_\prec$ be a region trimmed $\intfo01$-\wtg of region assignment $\reg$. 
Let $\Gcal_\preceq$ be a copy of $\Gcal_\prec$ where 
\begin{itemize}
\item every guard has been relaxed, i.e., every guard of the form $x>0$ and $x<1$ have been replaced by $x\geq 0$ or $x\leq 1$ respectively, for all $x\in\Xcal$.
\item useless guards (see the second point of \cref{relaxedreg}) have been removed.
\end{itemize}
Then $\Gcal_\preceq$ is a relaxed trimmed region \wtg, of region assignment $\reg$.
\end{lemma}

\subsection{Adding Resets to every Transition}
\label{sec-simpl}

\begin{restatable}{lemma}{allreset}
\label{thm-all-reset}
For any relaxed region trimmed kernel $\intff01$-\wtg $\Gcal$, such that $\Gcal$ has no requirement $x<1$ for any $x\in\Xcal$ then
there exists a relaxed region trimmed kernel $\intff01$-\wtg $\Gcal'$ of same value and verifying the same conditions such that every transition of $\Gcal'$ is a reset transition or a transition to the target location.
Furthermore, any transition of $\Gcal'$ with, for some clock $x$, a guard of the form $x=0$ or $x=1$, resets $x$.
\end{restatable}

See \cref{sec-proof-all-reset} for the proof.

%% file: parts/valueiteration.tex
\section{Value Iteration in Two-clock Kernel Games}
\label{sec-value-kernel}
In this section, we prove \cref{lem-linear-kernel}
using the value iteration paradigm (see \cite{alur2004,BGHM16,brihaye2022}):

Let $\Gcal=(L_\Minsf,L_\Maxsf,G,\Xcal,T,w \text{ or } \wout)$ be a trimmed region (kernel) \wtg.
In a \wtg, the \textbf{value iteration algorithm} builds, for each location $\ell$ and for all $k\geq 0$, a function $\opt{\ell}{k}: \Reel_{\geq 0}^\Xcal \to \Rbb$ such that
$\opt{\ell}{k}(\nu)$ is the value of the game started in $\ell$ with clock valuation $\nu$, where $\Minsf$ has to win in at most $k$ steps.
The $\opt{}{}$ functions are built inductively:
\begin{itemize}
\item $\opt{\ell}{0}$ is the constant $0$ function if $\ell \in G$ (or $\wout(\ell,\cdot)$ in the case of a kernel \wtg), or the constant $+\infty$ function otherwise.
\item for any $k\in \Nat$, $\opt{\ell}{k+1}$ is obtained from the $\opt{}{}$ functions at step $k$: if $\ell$ belongs to $\Minsf$ (resp.~$\Maxsf$), then
\[\opt{\ell}{k+1}(\nu)=\inf \text{(resp.~$\sup$) } \setof{\opt{\ell'}{k}((\nu+\delta)[X:=0])}{\ell \overset{C,X}{\rightarrow} \ell'\in T, \nu+\delta \models C}\,.\]
\end{itemize}

Note that $\opt{\ell}{k}(\nu)=\Value^{\leq k}(\ell,\nu)$ the value of the game started from configuration $\ell,\nu$ when $\Min$ must reach $G$ in at most $k$ steps.
Naturally, $\opt{\ell}{k+1}(\nu)\leq \opt{\ell}{k}(\nu)$ for all valuation $\nu$. If there exists $k$ such that, for all locations $\ell$, $\opt{\ell}{k+1}= \opt{\ell}{k}$, then the value iteration algorithm terminates.

In general, there is no termination guarantee.
However, if there exists $k$ such that $\opt{\ell}{k+1}= \opt{\ell}{k}$ for all $\ell$, then $\opt{\ell}{k}(\nu)=\Value_\Gcal(\ell,\nu)$.
 This means that the value of the \wtg is obtained even when considering plays of length at most $k$.
 
Here is an example where the value iteration algorithm does not terminate.

\begin{figure}
\centering
\includegraphics[scale=1]{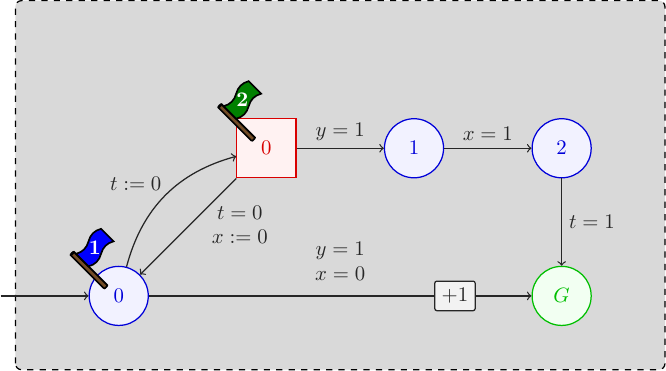}
\caption{ A \wtg where the value iteration algorithm does not terminate.
Blue circle locations belong to $\Minsf$, red square locations belong to $\Maxsf$, the green circle location $G$ is the target. The flags serve to easily refer to locations. The $+1$ label is a transition cost.}
\label{fig-3clock}
\end{figure}

\begin{example}
The \wtg $\Gcal$ in \cref{fig-3clock} is a $3$-clock, almost non-Zeno \wtg with a value of $1$. The kernel of $\Gcal$ contains only the cycle between \includegraphics[scale=.5]{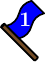} and \includegraphics[scale=.5]{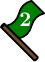}.
The cost of the output edge in \includegraphics[scale=.5]{flags/flag2.pdf} is $a+2\delta$ for a valuation $(x,y,t)=(\delta, a+\delta,0)$, whereas the cost of the output edge in \includegraphics[scale=.5]{flags/flag1.pdf} is exactly $1$, for a valuation $(x,y,t)=(0, 1,0)$.

An optimal strategy for $\Minsf$ is to loop in the kernel an arbitrary number of times: In this strategy, each time she enters \includegraphics[scale=.5]{flags/flag1.pdf} with valuation $(x,y,t)=(0,a,t)$, she should waits $\delta$ such that $2\delta = 1-a$ and enters \includegraphics[scale=.5]{flags/flag2.pdf}, then $\Maxsf$ can either reach $G$ with a cost of exactly $a+2\delta=1$, or return to \includegraphics[scale=.5]{flags/flag1.pdf}. When $\Minsf$ decides to end the game, she then picks $\delta=1-a$ instead. Then $\Maxsf$ chooses between going to $G$ at cost $1 + (1-a)$, or letting $\Minsf$ leave from \includegraphics[scale=.5]{flags/flag1.pdf} at cost $1$.
 Each time $\Minsf$ takes the cycle with delay $\delta$ such that $2\delta = 1-a$ , $y$ gets closer to $1$, thus minimizing the cost of picking 
 $\delta=1-a$ at some point. Thus $\Minsf$ has a strategy to reach $G$ with cost $>1$, but arbitrarily close to $1$ depending on how long she plays.
This entails that the value of this \wtg is obtained by considering arbitrarily long plays. Thus, the value iteration algorithm does not terminate on this example.
\end{example}

However, the value iteration algorithm terminates for two-clock kernel \wtgs.

\thmKernel*

\begin{proof}
Without loss of generality, let us assume that $\Gcal$ is a relaxed trimmed region kernel $\intff01$-\wtg where every transition is either a transition to a target location, or resets at least one clock (see \cref{lem-zero-one,thm-sim-spicy,thm-all-reset}). Furthermore, we assume that there is no transition to the initial location $i$.
\footnote{This can be done by making a copy of the initial location, such that all transition that should enter the initial location only enter the copy instead.}

For every location $\ell\not\in G \cup \{i\}$, the set of valuations $\regadh{\ell}$  is either $\setof{(0,y)}{y\in \intff01}$ or $\setof{(x,0)}{x\in \intff01}$, or the singleton $\{(0,0)\}$. Note that the $\opt{\ell}{}$ functions will be defined on $ \regadh{\ell}$ for any location $\ell$. This entails that the value iteration algorithm will mainly build $1$-dimensional functions. 

Let us highlight this observation with a subtle change of variable: Let $\Delta$ the \textbf{circular clock difference} of a valuation $\nu\in\regadh{\ell}$ be defined as: 
\[\Delta(\nu)=\begin{cases}
y \text{ when } \nu=(0,y)\text{ with $y\geq 0$,}\\
1-x \text{ when } \nu=(x,0) \text{ with $x>0$.}
\end{cases}\]

Now, for any  $\ell\not\in G \cup \{i\}$, for any $k\in\Nat$, let
$\Opt{\ell}{k}(\Delta(\nu))\eqdef\opt{\ell}{k}(\nu)$ for all $\nu\in \regadh{\ell}$. The function $\Opt{\ell}k$ is either defined on $\{0\}$ if $\regadh{\ell}=\{(0,0)\}$, or on $\intff01$.


\input{figures/fig-region-reset}
\input{figures/fig-region-reset-x}

The following observation motivates this change of variable:
In \cref{fig-rry,fig-rrx}, we consider a transition $t$, from locations $\ell$ to $\ell'$ such that the $\Opt{\ell}{}$ and $\Opt{\ell'}{}$ functions are defined on $\intff01$, with no $=0$ or $=1$ guards.
For a transition $(\ell,\nu) \overset{t,\delta}\rightarrow
(\ell',\nu')$, we then observe that $\Delta(\nu')\in [\Delta(\nu), 1]$ if $\nu(x)=0$ (\cref{fig-rry}), and $\Delta(\nu')\in [0,\Delta(\nu)]$ if $\nu(y)=0$ (\cref{fig-rrx}). Thus the variation of $\Delta$ only depends on the region of $\ell$, not $\ell'$. In other words, it does not matter which clocks are reset by a transition, just the number of clocks that are reset.

Let us now define the induction relation between the $\Opt{}{}$ functions.
Let $\ell$ be a location such that $\ell\not\in G\cup\{i\}$. Moreover, assume that $\ell$ belongs to $\Minsf$ (resp.~ $\Maxsf$ ).
For $k=0$, $\Opt{\ell}{k}=\Delta \mapsto +\infty$.
Consider $\Next(\ell)$ the set of outgoing transition from a location $\ell$ belonging to $\Minsf$ (resp.~ $\Maxsf$ ), with $\ell\neq i$. For each $t:\ell\overset{C,X}{\rightarrow} \ell'$, for any $k\in\Nat$, let us define a function $\Opt{t}{k}$ such that $\Opt{\ell}{k+1}=\min_{t\in\Next(\ell)} \Opt{t}{k}$ (resp.~ $\max$):

\begin{itemize}
\item if $\ell'$ is a goal location, then for all $k\in\Nat$,
	\begin{itemize}
	\item If $\regadh{\ell}=\setof{(0,y)}{y\in \intff01}$,
	then 
	\[\Opt{t}{k}=\Opt{t}{0}: \Delta \mapsto\underset{(\delta,\Delta+\delta)\models C}{\underset{0\leq \delta\leq 1-\Delta} {\inf}}\, \wout(\ell',(\delta,\Delta+\delta)[X:=0])\text{ (resp.~ $\sup$).}\]
	\item If $\regadh{\ell}=\setof{(x,0)}{x\in \intff01}$,
		then \[\Opt{t}{k}=\Opt{t}{0}: \Delta \mapsto\underset{(1-\Delta+\delta,\delta)\models C}{\underset{0\leq \delta\leq \Delta} {\inf}}\, \wout(\ell',(1-\Delta+\delta,\delta)[X:=0])\text{ (resp.~ $\sup$).}\]
	\item If $\regadh{\ell}=\{0,0\}$, 
	then $\Opt{t}{k}(0)=\Opt{t}{0}(0)= \underset{(\delta,\delta)\models C}{\underset{0\leq \delta\leq 1} {\inf}}\, \wout(\ell',(\delta,\delta)[X:=0])$ (resp.~ $\sup$).
	
	\end{itemize}
	
\item if $\ell'$ is not a goal location, then $X\neq\emptyset$ ($t$ resets one or two clocks).
	\begin{itemize}
	\item If $X=\{x,y\}$ then, since $\Gcal$ is almost trimmed, every valuation in $\regadh{\ell}$ can reach $(\ell',(0,0))$. Therefore $\Opt{t}{k}= \Delta \mapsto \Opt{\ell'}{k}(0)$.
	\item Otherwise, $\Opt{\ell'}{k}$ is defined on $\intff01$. Then:
		\begin{itemize}
		\item If $\regadh{\ell}=\setof{(0,y)}{y\in \intff01}$,
			\begin{itemize}
			\item If $(x=0)\in C$ (thus $X=\{x\}$) or $(y=1)\in C$ (thus $X=\{y\}$),\footnote{Note that since $\Gcal$ is almost trimmed, both guards cannot be in $C$ at the same time. Then $\Opt{t}{k}= \Delta \mapsto \Opt{\ell'}{k}(\Delta)$.} 
			then this forces a delay such that $t$ preserves $\Delta$. 
			\item Otherwise, as observed in \cref{fig-rry},
			\[\Opt{t}{k}: \Delta \mapsto\underset{\Delta\leq \Delta'\leq 1}{\inf}\, \Opt{\ell'}{k}(\Delta')\text{ (resp.~ $\sup$).}\]
			\end{itemize}
		\item Symetrically, if $\regadh{\ell}=\setof{(x,0)}{x\in \intff01}$,
		\begin{itemize}
			\item If $(y=0)\in C$ or $(x=1)\in C$, then $\Opt{t}{k}= \Delta \mapsto \Opt{\ell'}{k}(\Delta)$.
			\item Otherwise, as observed in \cref{fig-rrx},
			\[\Opt{t}{k}: \Delta \mapsto\underset{0\leq \Delta'\leq \Delta}{\inf}\, \Opt{\ell'}{k}(\Delta')\text{ (resp.~ $\sup$).}\]
			\end{itemize}
		\item If $\regadh{\ell}=\{(0,0)\}$,
		$\Opt{t}{k}(0)= \underset{0\leq \Delta\leq 1}\inf \Opt{\ell'}{k}(\Delta)$ (resp.~ $\sup$).
		\end{itemize}
	\end{itemize}

\end{itemize}

%

%

We call the functions $\Opt{t}{0}$, for all transitions $t$ to a goal location, the \textbf{projected output functions} of $\Gcal$.  
The projected output functions serve to initialize the value iteration algorithm on the $\Opt{}{}$ functions, as do $\wout(\ell,\cdot)$ functions for the value iteration algorithm on the $\opt{}{}$ functions. Observe that, since $\wout(\ell,\cdot)$ is piecewise linear and continuous for any $\ell\in G$, the projected output functions are continuous piecewise-linear functions.

\begin{figure}[ht]
\includegraphics[scale=.7]{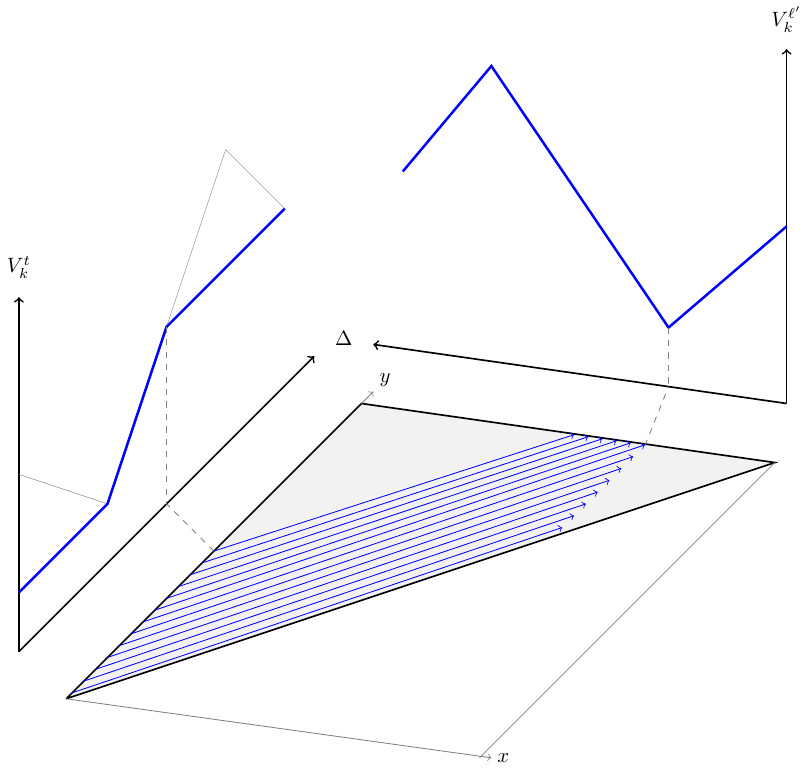}
\caption{$\Opt{}{k}$ induction relation, with $t$ a transition from a $\Minsf$ location $\ell$ where $x\leq y$. $t$ has no $=0$ or $=1$ guards, and applies $y:=0$.}
\label{fig-slash}
\end{figure}

Thus the functions $\Opt{\ell}{k}$ are by construction continuous piecewise-linear functions. Furthermore,
as can be seen in \cref{fig-slash}, $\Opt{t}{k}$ is obtained from $\Opt{\ell'}{k}$ by replacing $\Opt{\ell'}{k}$ on a finite number of intervals by constant functions $\Delta\mapsto c$ (while preserving continuity) where the constants $c$ are taken among local extremums of $\Opt{\ell'}{k}$.

Hence every linear piece of a function $\Opt{\ell}{k}$ is:
\begin{itemize}
\item either equal to some projected output function,
\item or of slope $0$, equal to some $z\in \mathbb{R}^+$, where $z$ is a local extremum of some function $\Opt{\ell}{k'}$ for $k'<k$. Hence, by induction, $z$ is either a local extremum of some projected output function, or a local extremum of the minimum or maximum of two projected output functions.
\end{itemize}

There is a finite number of such pieces, hence only a finite number of way they can be assembled to make a \emph{continuous} piecewise linear function on $\intff01$.
Thus there is a finite number of functions of this form.

Since the $\opt{}{}$ and $\Opt{}{}$ functions  decrease at each iteration ($\Opt{\ell}{k+1}(\Delta)\leq\Opt{\ell}{k}(\Delta)$), then the value iteration algorithm on the $\Opt{}{}$ functions terminates.

Furthermore, since no transition enters $i$, $\opt{i}{k}$ does not affect $\opt{\ell}{k+1}$ for any location $\ell$. Therefore, the value iteration algorithm on the $\opt{}{}$ functions aside from $\opt{i}{}$  stabilizes. Thus, if it terminates in $k$ steps for every location except $i$, then, adding $\opt{i}{}$, the value iteration algorithm terminates in at most $k+1$ steps.
\end{proof}

Note that all transformations in \cref{sec-simpl} only serve to make the termination argument of \cref{lem-linear-kernel} more visible. However, the value iteration algorithm on two-clock kernel \wtgs terminates even without these simplifications. Indeed, termination in $k$ steps entails that $\Minsf$ needs only to consider strategies that access a goal location in at most $k$ steps. Since the transformations described in \cref{sec-all-simpl} do not make arbitrarily long paths equivalent to one shorter path, then termination of the value algorithm on the transformed kernel \wtgs immediately implies termination on the original kernel \wtgs. This in turn entails that the value iteration algorithm terminates on the semi-unfolding $\tilde{\Gcal}$, thus on any two-clock \wtg with non-negative weight.

%% file: figures/fig-region-reset.tex
\begin{figure}[t]
\centering
\scalebox{.8}{
\begin{tikzpicture}

\def\xx{5}
\def\dd{2}
\def\xdd{3}
\def\ee{4}
\def\nn{.1}
\def\zz{7}

\fill[gray!10] (0, 0) -- (0,\xx) -- (\xx,\xx) -- (0,0);
\draw[gray, very thin] (0, 0) -- (\xx,0)-- (\xx, \xx);
\draw[->,gray, very thin] (\xx, 0) -- (\xx+.2, 0);
\draw[->,gray, very thin] (0, \xx) -- (0,\xx+.2);
\draw[dashed,gray, very thin] (-.6, \dd) -- (-1.5,\dd);
\draw[dashed,gray, very thin] (-.2, \dd) -- (0,\dd);
\draw[thick] (0, 0) -- (0,\xx) -- (\xx,\xx) -- (0,0);
\node[gray] at (\xx+.4,0){$x$};
\node[gray] at (0,\xx+.4){$y$};

\node (A) at (0,\dd){};
\node (B) at (\ee-\dd,\ee){};
\node (C) at (0, \ee){};

\draw [blue, very thick] (0,\dd) -- (\ee-\dd,\ee);
\draw [blue, very thick,dashed] (\ee-\dd,\ee) --(0,\ee) ;
\draw [gray, very thin] (\ee-\dd,\ee) -- (\xx-\dd,\xx);

  \foreach \y in {.2,.4,...,1.8, 2.2, 2.4, ..., \xdd}{

  \draw [gray,thin] (\y,\y+\dd) --(0,\y+\dd);

  }
  
  \draw[->] (-1.6, 0) -- (-1.6, \xx+.2);

\node at (-.4,\dd){$\nu$};
\node at (\ee-\dd+.7,\ee){$\nu+\delta$};
\node at (-.7,\ee+.2){$\nu+\delta$};
\node at (-.7,\ee-.2){$[x:=0]$};

\draw (-1.7,\dd) -- (-1.5,\dd);
\draw (-1.7,\ee) -- (-1.5,\ee);
\draw (-1.7,0) -- (-1.5,0);
\node at (-2,\dd){$\Delta$};
\node at (-2.3,\ee){$\Delta+\delta$};

\fill[gray!10] (0+\zz, 0) -- (\zz,\xx) -- (\xx+\zz,\xx) -- (0+\zz,0);
\draw[gray, very thin] (0+\zz, 0) -- (\xx+\zz,0)-- (\xx+\zz, \xx);
\draw[->,gray, very thin] (\xx+\zz, 0) -- (\xx+.2+\zz, 0);
\draw[->,gray, very thin] (0+\zz, \xx) -- (0+\zz,\xx+.2);
\draw[thick] (0+\zz, 0) -- (0+\zz,\xx) -- (\xx+\zz,\xx) -- (0+\zz,0);
\node[gray] at (\xx+.4+\zz,0){$x$};
\node[gray] at (0+\zz,\xx+.4){$y$};

\node (A) at (0+\zz,\dd){};
\node (B) at (\ee-\dd+\zz,\ee){};
\node (C) at (0+\zz, \ee){};

\draw [blue, very thick] (0+\zz,\dd) -- (\ee-\dd+\zz,\ee);
\draw [blue, very thick,dashed] (\ee-\dd+\zz,\ee) --(\ee-\dd+\zz,0) ;
\draw [gray, very thin] (\ee-\dd+\zz,\ee) -- (\xx-\dd+\zz,\xx);

  \foreach \y in {.2,.4,...,1.8, 2.2, 2.4, ..., \xdd}{

  \draw [gray,thin] (\y+\zz,\y+\dd) --(\y+\zz,0);

  }
  
\draw[->] (\xx+\zz, -1.2) -- (0+\zz, -1.2);

\node at (-.4+\zz,\dd){$\nu$};
\node at (\ee-\dd-.7+\zz,\ee){$\nu+\delta$};
\node at (\ee-\dd+\zz,-.2){$\nu+\delta$};
\node at (\ee-\dd+\zz,-.7){$[y:=0]$};
\draw[dashed,gray, very thin] (\xx+\zz-\dd, 0) -- (\xx+\zz-\dd,-1.5);
\draw[dashed,gray, very thin] (\ee+\zz-\dd, -.8) -- (\ee+\zz-\dd,-1.1);
\draw[black] (\xx+\zz-\dd,-1.1) -- (\xx+\zz-\dd,-1.3);
\draw[black] (\ee+\zz-\dd,-1.1) -- (\ee+\zz-\dd,-1.3);
\draw[black] (\xx+\zz,-1.1) -- (\xx+\zz,-1.3);

\node at (\xx+\zz-\dd,-1.6){$\Delta$};
\node at (\ee+\zz-\dd+.1,-1.6){$1-\delta>$};

\end{tikzpicture}
}
\caption{Evolution of $\Delta$ in a one-clock-reset transition, without guards $=0$ or $=1$, from region $\setof{(0,y)}{y\in \intff01}$.}
\label{fig-rry}

\end{figure}

%% file: figures/fig-region-reset-x.tex
\begin{figure}[t]
\centering
\scalebox{.8}{
\begin{tikzpicture}

\def\xx{5}
\def\dd{2}
\def\xdd{3}
\def\ee{4}
\def\nn{.1}
\def\zz{9}

\draw[gray, very thin] (0, 0) -- (0,\xx)-- (\xx, \xx);
\draw[->,gray, very thin] (\xx, 0) -- (\xx+.2, 0);
\draw[->,gray, very thin] (0, \xx) -- (0,\xx+.2);
\fill[gray!10] (0, 0) -- (\xx,0) -- (\xx,\xx) -- (0,0);
\draw[thick] (0, 0) -- (\xx,0) -- (\xx,\xx) -- (0,0);
\node[gray] at (\xx+.4,0){$x$};
\node[gray] at (0,\xx+.4){$y$};

\draw [blue, very thick] (\dd,0) -- (\ee,\ee-\dd);
\draw [blue, very thick,dashed] (\ee, \ee-\dd) --(\ee, 0);
\draw [gray, very thin] (\ee,\ee-\dd) -- (\xx,\xx-\dd);

  \foreach \y in {.2,.4,...,1.8, 2.2, 2.4, ..., \xdd}{

  \draw [gray,thin] (\y+\dd,\y) --(\y+\dd,0);

  }
  
  \draw[->]  (\xx,-1.2) -- (0,-1.2);

\node at (\dd,-.4){$\nu$};
\node at (\ee,\ee-\dd+.7){$\nu+\delta$};
\node at (\ee, -.3){$\nu+\delta$};
\node at (\ee, -.7){$[y:=0]$};
\draw[dashed,gray, very thin] (\dd,-.6) -- (\dd,-1.1);
\draw[dashed,gray, very thin] (\dd,-.3) -- (\dd,0);
\draw[dashed,gray, very thin] (\ee,-.9) -- (\ee,-1.1);

\draw (\dd, -1.1) -- (\dd, -1.3);
\draw (\ee,-1.1) -- (\ee,-1.3);
\draw (\xx,-1.1) -- (\xx,-1.3);
\node at (\dd, -1.5){$\Delta$};
\node at (\ee, -1.5){$\Delta-\delta$};

\fill[gray!10] (0+\zz, 0) -- (\xx+\zz,0) -- (\xx+\zz,\xx) -- (0+\zz,0);
\draw[gray, very thin] (0+\zz, 0) -- (\zz,\xx)-- (\xx+\zz, \xx);
\draw[->,gray, very thin] (\xx+\zz, 0) -- (\xx+.2+\zz, 0);
\draw[->,gray, very thin] (0+\zz, \xx) -- (0+\zz,\xx+.2);
\draw[thick] (0+\zz, 0) -- (\xx+\zz,0) -- (\xx+\zz,\xx) -- (0+\zz,0);
\node[gray] at (\xx+.4+\zz,0){$x$};
\node[gray] at (0+\zz,\xx+.4){$y$};

\draw [blue, very thick] (\dd+\zz,0) -- (\ee+\zz, \ee-\dd);
\draw [blue, very thick,dashed] (\ee+\zz, \ee-\dd) --(0+\zz, \ee-\dd) ;
\draw [gray, very thin] (\ee+\zz, \ee-\dd) -- (\xx+\zz, \xx-\dd);

  \foreach \y in {.2,.4,...,1.8, 2.2, 2.4, ..., \xdd}{

  \draw [gray,thin] (\y+\zz+\dd,\y) --(0+\zz,\y);

  }
  
\draw[<-] (-1.6+\zz,\xx) -- (-1.6+\zz, 0);

\node at (\dd+\zz,-.4){$\nu$};
\node at (\ee+\zz,\ee-\dd-.7){$\nu+\delta$};
\node at (-.7+\zz,\ee-\dd+.2){$\nu+\delta$};
\node at (-.7+\zz,\ee-\dd-.2){$[x:=0]$};
\draw[dashed,gray, very thin] (\zz,\xx-\dd) -- (-1.5+\zz,\xx-\dd);

\draw[black] (-1.5+\zz,\xx-\dd) -- (-1.7+\zz,\xx-\dd);
\draw[black] (-1.5+\zz,\ee-\dd) -- (-1.7+\zz,\ee-\dd);
\draw[black] (-1.5+\zz,0) -- (-1.7+\zz,0);

\node at (-1.9+\zz,\xx-\dd){$\Delta$};
\node at (-2.3+\zz,\ee-\dd){$\delta<\Delta$};

\end{tikzpicture}
}
\caption{Evolution of $\Delta$ in a one-clock-reset transition, without guards $=0$ or $=1$, from region $\setof{(x,0)}{x\in \intff01}$.}
\label{fig-rrx}
\end{figure}

%% file: parts/appendix.tex
\section{Proofs of Section \ref{sec-simpl}}
\subsection{Proof of Lemma \ref{thm-sim-spicy}}
\label{sec-app}

\begin{definition}[Simulation and bisimulation]

Let $\mathcal{G}=(L_\Minsf,L_\Maxsf,G,\mathcal{X},T,w)$ and $\mathcal{G}'=(L'_\Minsf,L'_\Maxsf,G',\mathcal{X},T',w')$ be \wtgs. Then a relation $R\subseteq (L_\Minsf\times \Rbb^\Xcal\times L'_\Minsf\times \Rbb^\Xcal) \cup (L_\Maxsf\times \Rbb^\Xcal\times L'_\Maxsf\times \Rbb^\Xcal)$ is a \textbf{simulation relation} when
\[(\ell,\nu)\,R\,(\ell',\nu') \implies \begin{cases}
\medskip
\text{$\ell\in G$ and $\ell'\in G'$,}\\
\text{or, for all $t\in T$, $\delta\in\Rbb$, there exist $t'$ and $\delta'\in\Rbb$ such that if }\\
 \text{$(\ell,\nu)\xrightarrow{\delta,t}(l_2,\nu_2)$ then $(\ell',\nu')\xrightarrow{\delta',t'}(\ell'_2,\nu'_2)$ where $(\ell_2,\nu_2)\,R\,(\ell'_2,\nu'_2)$.}
\end{cases}\]


$R$ is a \textbf{bisimulation} if $R$ and its converse are both simulations.
\end{definition}

\begin{example}
 The region relation $R$ between configurations of a same \wtg, where $(l_1,\nu_1)R(l_2,\nu_2)$ iff $l_1=l_2$ and $\nu_1,\nu_2$ are in the same region, is a bisimulation.
\end{example}

A (bi)simulation relation $R$ can be extended to runs:
$\rho_1 \,R\, \rho_2$ when $|\rho_1|=|\rho_2|$, and  $\rho_1^\Ccal(n) \,R\, \rho_2^\Ccal(n)$ for all $n<|\rho_i|$.

\begin{lemma}
\label{lem-bisim-WTG}
Let $\Gcal_1$ and $\Gcal_2$ be two \wtgs. 
Let $R$ be a bisimulation relation between $\Gcal_1$ and $\Gcal_2$, such that $\Gcal_1$ and $\Gcal_2$ start from configurations $c_1 \,R\, c_2$. Assume that $\Value_{\Gcal_i}\neq +\infty$ for $i\in \{1,2\}$. 
Then 

$$|\Value_{\Gcal_1}-\Value_{\Gcal_2}|\leq \sup
\enstq{|\weight(\rho_1)-\weight(\rho_2)|}{\begin{array}{c}
		\rho_1,\rho_2 \text{ plays of } \Gcal_1,\Gcal_2\\ 
		\rho_1 \,R\, \rho_2\\
		\end{array}}
\;.$$

\end{lemma}
\begin{proof}
For all $\sigma_1^\Mimin$ a strategy for Min on $\Gcal_1$ and $\sigma_2^\Mamax$ a strategy for Max on $\Gcal_2$, there exists  $\sigma_2^\Mimin$ and $\sigma_1^\Mamax$ strategies for Min and Max on $\Gcal_2$ and $\Gcal_1$ respectively, such that $\rho_1 \,R\, \rho_2$ for $\rho_i=\play(c_i,\sigma_i^\Mimin,\sigma_i^\Mamax)$.
(In any Min location, $\sigma_2^\Mimin$ simulates $\sigma_1^\Mimin$ following $R$. Similarly in any Max location, $\sigma_1^\Mamax$ simulates $\sigma_2^\Mamax$ following $R$.)

Note that $\sigma_1^\Mamax, \sigma_2^\Mimin$ and $\rho_1,\rho_2$ are functions of $\sigma_1^\Mimin, \sigma_2^\Mamax$. However, we omit these arguments to lighten notations.

Since the set of strategies obtained by such a simulation is included in the set of all strategies,
\begin{align*}
\sup_{\sigma_2^\Mamax} \weight(\rho_1) \leq \sup_{\sigma^\Mamax} V_{\Gcal_1}(\sigma_1^\Mimin, \sigma^\Mamax)\text{ for a fixed $\sigma_1^\Minsf$,}\\
\shortintertext{and}
\inf_{\sigma_1^\Mimin} \weight(\rho_2) \geq \inf_{\sigma^\Mimin} V_{\Gcal_2}(\sigma^\Mimin, \sigma_2^\Mamax)\text{ for a fixed $\sigma_2^\Maxsf$.}
\end{align*}

Therefore \[\inf_{\sigma_1^\Mimin} 
\sup_{\sigma_2^\Mamax}
\weight(\rho_1)\leq 
\inf_{\sigma^\Mimin}\sup_{\sigma^\Mamax} \weight(\play_{\Gcal_1}(c_1,\sigma^\Mimin, \sigma^\Mamax))= \Value_{\Gcal_1}\]
and
\begin{align*}
\inf_{\sigma_1^\Mimin}
\sup_{\sigma_2^\Mamax} 
\weight(\rho_2)
&\geq
\sup_{\sigma_2^\Mamax} 
\inf_{\sigma_1^\Mimin}
\weight(\rho_2)\\
&\geq 
\sup_{\sigma^\Mamax}\inf_{\sigma^\Mimin} \weight(\play_{\Gcal_2}(c_2,\sigma^\Mimin, \sigma^\Mamax))= \Value_{\Gcal_2}
\end{align*}

Therefore $\Value_{\Gcal_2}-\Value_{\Gcal_1}\leq 
\sup\setof{\weight(\rho_2)-\weight(\rho_1)}{\rho_1 \,R\, \rho_2}\;.$

%

Reasoning in mirror, one can obtain
\[\Value_{\Gcal_1}-\Value_{\Gcal_2}\leq 
\sup\setof{\weight(\rho_1)-\weight(\rho_2)}{\rho_1 \,R\, \rho_2}\;.\]

%
and combine to conclude.

%
%
\end{proof}

Now let us define the following relation between valuations:
For $\epsilon>0$, two valuations $\nu,\nu'$ are \textbf{$\epsilon$-neighbours} if there exists $\epsilon_1,\epsilon_2\geq 0$ such that $\epsilon_1+\epsilon_2<\epsilon$, and for any clock $x\in\Xcal$, $\nu(x)-\nu'(x)\in[-\epsilon_1,\epsilon_2]$.

Let $\Gcal_\prec=(L_\Minsf^\prec, L_\Maxsf^\prec,G,\Xcal,T^\prec, \wout^\prec)$ be a trimmed region weighted timed game. We consider $\Gcal_\preceq = (L_\Minsf^\preceq, L_\Maxsf^\preceq,G,\Xcal,T^\preceq,\wout^\preceq)$ a copy of $\Gcal_\prec$ where every strict guard has been relaxed into a strict-or-equal guard.

\begin{lemma}
\label{lem-bisim-spicy-all}

For any configurations $(\ell_\prec,\nu_\prec)$ in $\Gcal_\prec$, 
and $(\ell_\preceq,\nu_\preceq)$ in $\Gcal_\preceq$, 

Let $(\ell_\prec,\nu_\prec) \,R_\epsilon\, (\ell_\preceq,\nu_\preceq)$ \tiff 
\begin{itemize}
\item $\ell_\prec=\ell_\preceq$ 
\item $\nu_\prec\sqsubset \reg(\ell_\prec)$ and $\nu_\preceq\sqsubset \overline{\reg(\ell_\prec)}$
\item $\nu_\prec$ and $\nu_\preceq$ are $\epsilon$-neighbours
\end{itemize}
Then $R_\epsilon$ is a bisimulation relation.
\end{lemma}

\begin{proof}
First observe that one side of the bisimulation is easier than the other one, since transitions in $\Gcal_\preceq$ are more permissive.
Therefore we will detail only the other side:

Let $t:\ell_1\xrightarrow{C,X} \ell_2$. Let $\overline{C}$ be the relaxation of $C$. Let $\nu_1$ be a valuation belonging to $\reg(\ell_1)$ in $\Gcal_\prec$ and let $\nu_1'$ be an $\epsilon$-neighbour of $\nu_1$. By definition, there are $\epsilon_1,\epsilon_2\geq 0$ such that $\epsilon_1+\epsilon_2<\epsilon$ and for any clock $x\in\Xcal$, $\nu_1(x)-\nu_1'(x)\in[-\epsilon_1,\epsilon_2]$.

Let $\delta'\geq 0$ such that $\nu_1'+\delta' \models \overline{C}$ and let $\nu_2'=(\nu_1'+\delta')[X:=0]$. Then let us show that there exists $\delta\geq 0$ such that $\nu_1 + \delta \models C$,
and
$(\ell_2,\nu_2)R_\epsilon (\ell'_2,\nu'_2)$ with $\nu_2=(\nu_1+\delta)[X:=0]$. To prove that, one only needs to show that $\nu'_1+\delta'$ is in the adherence of the region of $\nu_1+\delta$, and that $\nu_1+\delta$ and $\nu'_1+ \delta'$ are $\epsilon$-neighbours.

Consider the interval
\[\Delta=\enstq{\delta\geq0}{\begin{array}{c}
		\nu_1+\delta\models C\\ 
		\nu_1+\delta\sqsubset r\text{ s.t. }\nu_1'+\delta' \sqsubset \overline{r}\\
\end{array}}\]
Since $\Gcal_\prec$ is region trimmed, there exists $\delta$ such that $\nu_1+\delta\models C$, and every valuation of the form $\nu+\delta$ that satisfy $C$ with $\nu\sqsubset \reg(\ell_1)$ belong to a unique region $r_t$. Since $\nu'_1\sqsubset \regadh{\ell_1}$ and $\nu'_1 + \delta' \models \overline{C}$, $\nu'_1 + \delta' \sqsubset \overline{r_t}$.
Therefore, $\Delta$ is not empty.

If $\delta' \in \Delta$, then $(\nu_1+\delta')[X:=0]$ and $\nu_1'+\delta'$ are $\epsilon$-neighbours.
Otherwise, when $\delta'\not\in \Delta$, it entails that $\Delta$ is constrained by some guards in $C$. These guards can be of the form $x=1$, $x>0$, $x<1$ or $x\leq 1$ for any $x\in\Xcal$
\footnote{There cannot be any guard of the form $x=0$ in $C$, otherwise $\Delta=\{\delta'\}=\{0\}$. 
}
.

		\begin{itemize}
			\item If there is a guard $x=1$ in $C$ for some $x\in\Xcal$, then $\Delta$ is the singleton $\{1-\nu_1(x)\}$. Indeed $\Delta$ is not empty and $\delta\eqdef 1-\nu_1(x)$ is the only delay that can satisfy the guard $x=1$. Note that, for the same reason, $\delta'=1-\nu'_1(x)$. 
 By definition of $\Delta$, $\nu_2=(\nu_1+\delta)$ and $\nu_2'$ are in the same region. Moreover, for all $y\in\Xcal$, if $y\in X$, then $\nu_2(y)=\nu_2'(y)=0$ and if $y\in\Xcal\setminus X$, then 
			$
			(\nu_1(y) + \delta)-(\nu_1'(y) + \delta' \in\intff{-(\epsilon_1+\delta'-\delta)}{\epsilon_2-(\delta'-\delta)}
			$.
			Now, $\delta'-\delta=\nu_1(x)-\nu_1'(x)\in\intff{-\epsilon_1}{\epsilon_2}$ 
			entails that $\epsilon_2'\eqdef \epsilon_2-(\delta'-\delta)\geq 0$, and $\epsilon_1'\eqdef \epsilon_1+(\delta'-\delta)\geq 0$. Finally, $\epsilon_1'+\epsilon_2'=\epsilon_1+\epsilon_2<\epsilon$.
			\end{itemize}
Now assume that there is no such guards in $C$:
			\begin{itemize}
		\item If there is a guard $x>0$ in $C$ for some $x\in\Xcal$
		such that $\nu_1+\delta'\not\models (x>0)$ but $\nu_1' + \delta'\models (x>0)$ then $\nu_1'(x)+\delta'=0$.
		Furthermore, since $\Gcal_\prec$ is trimmed, $\nu_1(x)=0$. 
Therefore, $\Delta$ is an interval of the form $]0,\dots]$.
Pick $\delta\in \Delta$ such that $0<\delta<\epsilon-(\epsilon_1+\epsilon_2)$.
Let
$\epsilon_2'\eqdef \epsilon_2-(\delta'-\delta)=\epsilon_2+\delta$, and $\epsilon_1'\eqdef \min\{0,\epsilon_1+(\delta'-\delta)\}\geq0$. Then $\epsilon_1'+\epsilon_2'\leq \epsilon_1+\epsilon_2 + \delta <\epsilon$, hence $\nu_1+\delta$ and $\nu_1'+\delta'$ are $\epsilon$-neighbourss.	
\item If there is a guard $x< 1$ in $C$ (resp. $x\leq 1$) for some $x\in\Xcal$,
such that $\nu_1'+\delta'\models (x\leq 1)$ but $\nu_1+\delta' \not
\models (x< 1)$ (resp. $x\leq 1$), then $\Delta<\delta'$.
If $\delta' + \nu_1'(x) - \nu_1(x)\in \Delta$ then pick $\delta= \delta' + \nu_1'(x) - \nu_1(x)$.
			Since $\delta'-\delta=\nu_1(x)-\nu_1'(x)\in\intff{-\epsilon_1}{\epsilon_2}$, $\nu_1+\delta$ and $\nu_1'+\delta'$ are $\epsilon$-neighbours.

Otherwise 
pick $\delta\in\Delta$ such that
$ 1 -\epsilon + (\epsilon_1+\epsilon_2) <\nu_1(x)+\delta < 1$.
Then $\nu_1(x)- \nu_1'(x)<\delta'-\delta< \nu_1(x)- \nu_1'(x) + \epsilon-(\epsilon_1+\epsilon_2)$, with $\nu_1(x)- \nu_1'(x)
\in \intff{-\epsilon_1}{\epsilon_2}$.
Let
$\epsilon_2'\eqdef \min(0,\epsilon_2-(\delta'-\delta))\geq 0$ and $\epsilon_1'\eqdef \epsilon_1+(\delta'-\delta)\geq0$. Then $\epsilon_1'+\epsilon_2'<\epsilon$, hence $\nu_1' + \delta'$ is an $\epsilon$-neighbour of $\nu_1+\delta$.\qedhere
			
\end{itemize}

\end{proof}

We are now ready to prove \cref{thm-sim-spicy}.

\thmspicy*
\begin{proof}

$\Gcal_\prec$ and $\Gcal_\preceq$ start from the same configuration, thus their first configurations are in bisimulation $R_\epsilon$ for any $\epsilon>0$.

Therefore, according to \cref{lem-bisim-WTG},
\[|\Value_{\Gcal_\prec}-\Value_{\Gcal_\preceq}|\leq \sup
\enstq{|\weight(\rho_\prec)-\weight(\rho_\preceq)|}{\begin{array}{c}
		\rho_\prec,\rho_\preceq \text{ plays of } \Gcal_\prec,\Gcal_\preceq\\ 
		\rho_\prec \,R_\epsilon\, \rho_\preceq\\
		\end{array}}
\;.\]

In a kernel, for any finite run $\rho$ of length $n$, 
$\weight(\rho)=\wout(\rho^C(n))$. Let $s$ be the maximal slope (in absolute value, in one variable) in functions $\wout(\ell,\cdot)$. Then 

\begin{align*}
\sup&\setof{|\weight(\rho_1)-\weight(\rho_2)|}{\rho_1 \,R_\epsilon\, \rho_2}\\ 
&\leq \sup \setof{|\wout(\ell,\nu_1)-\wout(\ell,\nu_2)|}{\ell \in G, 
(\ell,\nu_1)\,R_\epsilon\, (\ell,\nu_2)}\\
&\leq \epsilon\cdot s \cdot |\Xcal|\qedhere
\end{align*}
Conclude with \cref{lem-bisim-WTG}.

\end{proof}

\subsection{Proof of Lemma \ref{thm-all-reset}}
\label{sec-proof-all-reset}

\begin{lemma}
\label{lem-no-mixed-gadgets}
For $A,B=\Minsf,\Maxsf$ or $\Maxsf,\Minsf$.
In a relaxed region kernel $\intff01$-\wtg, consider a transition $t:\ell_A \overset{C,X}{\rightarrow} \ell_B$ between $\ell_A$ a location belonging to Player $A$, and $\ell_B$ a location belonging to Player $B$, such that $C$ has no guard of the form $x=0,x=1,x<1$, and $X=\emptyset$. Then adding the guard $x=1$ to $C$ for some $x\in X^\uparrow_{\ell_A}$ does not change the value. 
\end{lemma}
\begin{proof}
Pick some $x\in X^\uparrow_{\ell_A}$.
By picking a delay $\delta<1-\nu(x)$, Player $A$ offers Player $B$ more options than if they picked $\delta=1-\nu(x)$. From the perspective of $B$, if a larger delay is advantageous, then they can take it from $\ell_B$ at cost $0$. Hence it is optimal for either $A$ or $B$ to pick the largest delay possible, i.e. $\delta=1-\nu(x)$. However, since $w(\ell_A)=w(\ell_B)=0$, forcing $A$ to take a delay in $\ell_A$ which would have been taken in $\ell_B$ by Player $B$ given the chance does not change the value.
Thus, restricting $A$ to strategies which, when choosing $t$ from a valuation $\nu$, choose a delay $\delta$ such that $(\nu+\delta)(x)=1$ for all $x\in X^\uparrow_{\ell_A}$ does not change the value of the \wtg.
\end{proof}
\allreset*
\begin{proof}
Let $\Gcal=(L_\Minsf, L_\Maxsf,G,\mathcal{X},T,w,\wout)$. 
We assume that $\Maxsf$ does not have full control over any cycle (\ie in any cycle there is a $\Minsf$ location from which she can decide to leave the cycle). Indeed, if $\Maxsf$ could reach such a cycle, the Value of $\Gcal$ would be $+\infty$.
Furthermore, let us assume that
there is no $\Minsf$ self-loop (a transition $\ell\in L_\Minsf\overset{C,X}\rightarrow\ell$) with $X=\emptyset$ in $\Gcal$: it does not make strategic sense for $\Minsf$ to take such a loop, hence they can be deleted without change in value.

Let us transform $\Gcal$ through the following operations.
For any $t:\ell_1 \overset{C,X}{\rightarrow} \ell_2$ of $\Gcal$ such that $\ell_2$ is not a target location, and $X=\emptyset$ and $C$ has no guards of the form $x=1$ for any clock $x$:
\begin{itemize}
\item If $C$ has a $x=0$ requirement for some clock $x$, then add $x$ to $X$.
\item If $C$ has no $x=0$ requirement for all clock $x$, and $\ell_1$ and $\ell_2$ belong to the same player, then remove $t$
and, for any $t':\ell_2\to \ell_3$ with $\ell_3\neq \ell_1$,
\footnote{
Adding a transition in the case $\ell_3=\ell_1$ would create a self-loop: $\Minsf$ has no use for self-loops without reset, and we assume that $\Maxsf$ has full control over no cycle, so this situation never happens when $\ell_1$ belongs to $\Maxsf$.}
 create a transition $t':\ell_1\to\ell_3$ such that $C(t'')=C(t)\cup C(t')$

\item If $C$ has no $x=0$ requirement for all clock $x$, and $\ell_1$ and $\ell_2$ do not belong to the same player, then let us add a $x=1$ requirement to $C$ where $x\in X^\uparrow_{\ell_1}$. According to \cref{lem-no-mixed-gadgets}, it does not change the value.
\end{itemize}

After these transformations, every transition without reset in $\Gcal$ is either a transition to a target location, or has a $x=1$ guard for some $x\in\Xcal$.

Then let us build $\Gcal'=(L'_\Minsf, L'_\Maxsf,G',\mathcal{X},T',w',\wout')$ a kernel \wtg
where:
\begin{itemize}
\item For any player $\Player$, let $L'_\Player=L_\Player \cup \setof{\ell_\downarrow}{\ell \in L_\Player}$.
\item Let $G'= G \cup \setof{\ell_\downarrow}{\ell \in G}$.
\item Start from $T'=\emptyset$.
For any $t:\ell\overset{C,X}{\to}\ell'$ in $T$, 
	\begin{itemize}
	\item if $\ell'\in G$ then for all $X\subseteq \upclock{\ell}$ add $t':\ell\overset{C,X}{\to} \ell'$ 
and	$t':\ell_\downarrow\overset{C_\downarrow,X}{\to} \ell'_\downarrow$ 
	to $T'$, where $C_\downarrow$ is $C\cup \setof{x=0}{x\in \upclock{\ell}}$ deprived of guards $(x=1)$ for all $x\in \upclock{\ell}$. Note that here $\upclock{\ell}$ is defined according to $\reg(\ell)$ the region assignment in the relaxed trimmed region \wtg $\Gcal$.
	\item if $\ell'\not\in G$ and $X=\emptyset$ then $(x=1)\in C$ for some $x\in \upclock{\ell}$.
	Then add $t':\ell \overset{C,\upclock{\ell}}{\to} \ell'_\downarrow$ 
	and $t'':\ell_\downarrow\overset{C_\downarrow,\upclock{\ell}}{\to} \ell'_\downarrow$ to $T'$.
	\item if $\ell'\not\in G$ and $X= \upclock{\ell}$, then
	add $t':\ell \overset{C,X}{\to} \ell'$  
	and $t'':\ell_\downarrow \overset{C_\downarrow,X}{\to} \ell'$
	to $T'$, such that $X'=\upclock{\ell}$ and $C'=C \cup \setof{x=0}{x\in X_1} \setminus \setof {x=1}{x\in X_1}$.
	\item if $\ell'\not\in G$ and $X\neq\emptyset$, then
	add $t':\ell \overset{C,X}{\to} \ell'$  
	and $t'':\ell_\downarrow \overset{C',X'}{\to} \ell'_\downarrow$
	to $T'$ such that $X'=\upclock{\ell}$ and $C'=C \cup \setof{x=0}{x\in X_1} \setminus \setof {x=1}{x\in X_1}$.
	\end{itemize}
\item For all $\ell\in G$, let $\wout'(\ell,\nu)=\wout(\ell,\nu)$, $\wout'(\ell_\downarrow,\nu)=\wout(\ell,\nu')$
where $\nu'(x)=\nu(x)$ for all $x\not\in X$, and $\nu'(x)=1$ otherwise.

\end{itemize}

Intuitively, in $\Gcal$, when one or several clocks are made to reach $1$ by a guard, there will usually be some urgent transitions taken until all these clocks have been reset. In $\Gcal'$, those clocks are immediately reset, and control moves to a location $\ell_\downarrow$. All paths leaving $\ell_\downarrow$  have $(x=0)$ conditions (for all $x$ that have been reset in $\Gcal'$ but not in $\Gcal$) to guarantee urgency. 
A valuation $\nu'$ in a location $\ell_\downarrow$ in $\Gcal'$ is thus equivalent to a valuation $\nu$ in $\ell$ in $\Gcal$ iff $\nu(x)=\nu'(x)$ if $x\not\in \upclock{\ell}$, and $\nu(x)=1$ and $\nu'(x)=0$ for $x\in \upclock{\ell}$.

$\Gcal'$ is a relaxed region \wtg with region assignment $\reg'(\ell)=\reg(\ell)$ and $\reg'(\ell_\downarrow)=(X_0\cup X_p, X_1,\dots,X_{p-1})$ when $\reg(\ell)=(X_0, X_1,\dots,X_{p})$.
From there, trim $\Gcal'$ to obtain a relaxed trimmed region \wtg.

\end{proof}